\documentclass[11pt]{article}
\usepackage[margin=1in]{geometry}
\usepackage{amsmath,amssymb,amsfonts,amsthm}
\usepackage{graphicx,xcolor,array,url,ifthen}
\RequirePackage{tgtermes}
\RequirePackage{newtxtext}
\RequirePackage{newtxmath}
\RequirePackage{bm}
\RequirePackage{endnotes}

\usepackage{algorithm}
\usepackage{algpseudocode}
\usepackage{tikz}

\usepackage{natbib}
 \bibpunct[, ]{(}{)}{,}{a}{}{,}%

\newtheoremstyle{TH}{\topsep}{\topsep}{\itshape}{}{\bfseries}{.}{ }{}
\newtheoremstyle{EX}{\topsep}{\topsep}{\normalfont}{}{\bfseries}{.}{ }{}
\theoremstyle{TH}
\newtheorem{theorem}{Theorem}
\newtheorem{lemma}{Lemma}
\newtheorem{proposition}{Proposition}
\newtheorem{corollary}{Corollary}
\newtheorem{conjecture}{Conjecture}
\theoremstyle{EX}
\newtheorem{definition}{Definition}
\newtheorem{example}{Example}

\newcommand{\Halmos}{\ensuremath{\square}}
\newcommand{\SingleSpacedXII}{}  
\newenvironment{APPENDICES}{\appendix}{}
\newcommand{\shimnote}[1]{\if\relax\detokenize{#1}\relax\else
  \par\smallskip{\footnotesize\noindent\textit{Note.}\enskip #1\par}\fi}
\newcommand{\FIGURE}[3]{{\centering #1\par}\caption{#2}\shimnote{#3}}
\newcommand{\TABLE}[3]{\caption{#1}{\centering #2\par}\shimnote{#3}}

\usepackage{enumitem}   
\usepackage{booktabs}   
\usepackage{longtable}  
\usepackage{array}      
\usepackage{pdflscape}  

\setlistdepth{9}
\renewlist{itemize}{itemize}{9}
\setlist[itemize,1]{label=$\bullet$}
\setlist[itemize,2]{label=$\circ$}
\setlist[itemize,3]{label=$-$}
\setlist[itemize,4]{label=$\ast$}
\setlist[itemize,5]{label=$\cdot$}
\setlist[itemize,6]{label=$\diamond$}
\setlist[itemize,7]{label=$\bullet$}
\setlist[itemize,8]{label=$\circ$}
\setlist[itemize,9]{label=$-$}
\renewlist{enumerate}{enumerate}{9}
\setlist[enumerate,1]{label=(\arabic*)}
\setlist[enumerate,2]{label=(\alph*)}
\setlist[enumerate,3]{label=(\roman*)}
\setlist[enumerate,4]{label=(\Alph*)}
\setlist[enumerate,5]{label=(\arabic*)}
\setlist[enumerate,6]{label=(\alph*)}
\setlist[enumerate,7]{label=(\roman*)}
\setlist[enumerate,8]{label=(\Alph*)}
\setlist[enumerate,9]{label=(\arabic*)}

\usetikzlibrary{decorations.pathreplacing,positioning,calc}

\providecommand{\coloneqq}{\mathrel{:}=}

\newcommand{\OPT}{\mathrm{OPT}}
\newcommand{\OPTp}{\mathrm{OPT}^{+}}
\newcommand{\cross}{\mathrm{cross}}
\newcommand{\lo}{\mathrm{lo}}
\newcommand{\hicmd}{\mathrm{hi}}
\newcommand{\goff}{\omega}
\newcommand{\kk}{\kappa}
\def\vo[#1]{\underline{v}(#1)}

\makeatletter
\renewenvironment{proof}[1][Proof]{%
  \par\addvspace{\medskipamount}%
  \noindent{\itshape #1.}\enskip\ignorespaces}%
  {\unskip\nobreak\hfill\Halmos\par\addvspace{\medskipamount}}%
\providecommand{\qedhere}{}
\makeatother

\theoremstyle{EX}

\begin{document}

\title{Retrieval Sequencing in a Compact Storage System with Side Access}

\author{Ya\u{g}mur G{\"u}l\\[2pt]
\small Institute of Sustainable Logistics and Mobility, University of Hamburg
Business School, Hamburg, Germany\\
\small \texttt{yagmur.guel@uni-hamburg.de}}
\date{}

\maketitle

\begin{abstract}

In a compact storage system with side access, bins are arranged in
adjacent vertical stacks. In each retrieval cycle, bins blocking access
are lifted to create a temporary tunnel through which a robot extracts
one requested bin from the side. The lifted bins are then lowered again,
and the energy expenditure of the cycle equals the number of lifted bins.
Given a pick list known in advance, we seek a retrieval order that
minimizes total energy expenditure. For every number of stacks, we prove
that the total energy separates into an instance constant and a pairwise
penalty, and that every energy-optimal order is laminar within each stack.
For two stacks, we prove that the retrieval sequencing problem is solvable
in polynomial time.
For three stacks, we give an exact polynomial-time reduction to a merge
problem and characterize which merge-problem instances arise from the
storage model. We define a recursive class of retrieval orders and
construct a polynomial-time dynamic program that minimizes energy
expenditure over this class. Finally, we quantify the additional energy
savings from batching, which allows several requested bins to be retrieved
in one cycle.
\end{abstract}


\section{Introduction}\label{sec:intro}

Compact storage systems eliminate the aisles of conventional warehouses.
Bins are arranged in vertical stacks placed directly next to one another,
and autonomous robots (bots) store and retrieve them. The bins requested
for retrieval are called target bins, or simply targets. In the side-access
system introduced by \citet{fauve2022storage}, a target bin is retrieved by
lifting all bins that block access, creating a temporary tunnel through
which a bot extracts the target from the side.
A retrieval cycle consists of lifting the blocking bins, extracting the
target bin, and lowering the lifted bins again. The energy expenditure
of the cycle equals the number of lifted bins.

We consider a pick list of target bins known before retrieval begins,
with exactly one target bin retrieved in each cycle. The objective is to
determine a retrieval order that minimizes total energy expenditure.
Earlier retrievals change the number of bins in the stacks and the
heights of the remaining target bins, affecting both their accessibility
and the energy required to retrieve them. We therefore investigate the
structure of optimal retrieval orders and how the number of stacks
affects the computational complexity of finding them.

\paragraph{Previous work.}
Research on compact storage includes puzzle-based
or live-cube systems \citep{gue2007puzzle,zaerpour2017small} and robotic
compact storage and retrieval systems
\citep{tappia2017modeling,zou2018operating}. AutoStore is a prominent
example of the latter class. \citet{meller2023} describes its principal
components and identifies the interacting requirements and decision
variables governing system design. Its robots access bins from above,
whereas the system studied here creates a temporary tunnel for side access.
Within the puzzle-based literature, \citet{yalcin2019single}
study single-item retrieval on a sliding grid. Related structural
assumptions and computability considerations for automated compact
storage are discussed by \citet{fauve2024optimierung}. For a comprehensive
overview of these and other types of storage systems, see
\citet{fauve2026review}, as well as \citet{azadeh2019robotized} and
\citet{boysen202550}.

Complexity boundaries governed by stack parameters also arise in
adjacent stacking models. \citet{lehnfeld2014loading} survey known
complexity results for stack loading, unloading, and premarshalling.
For the Block Retrieval Problem, \citet{deMeloDaSilva2018block} prove
NP-hardness when the number of stacks is part of the input and
polynomial-time solvability for every fixed number of stacks. For
storage loading with stacking constraints, \citet{bruns2016complexity}
obtain polynomial-time algorithms at stack capacity two and strong
NP-completeness for several variants in which each stack has capacity
three. Unlike these models, retrieval costs here arise from temporarily
lifting bins rather than relocating them between stacks.

\citet{SACRP-R1} initiate the algorithmic study of this recently introduced
storage system through the Side-Access Compact Retrieval Problem (SACRP),
which allows several targets to be retrieved in one cycle. They establish
strong NP-completeness and develop exact solution methods. We study the
single-target case to understand the structure and complexity of the
underlying sequencing problem. This removes the grouping decision while
retaining the effects of earlier retrievals on accessibility and energy
expenditure. Determining when this basic case admits polynomial-time
optimization provides a starting point for developing efficient retrieval
policies for the system.

\paragraph{Our Contribution.}
We express total energy expenditure as an instance-dependent constant and
a pairwise ordering penalty, and prove within-stack laminarity for any
number of stacks. The lower targets retrieved before a given target form
a consecutive block beginning with the next target below it. This
restriction remains valid despite interactions between stacks, allowing
algorithms to exclude all orders that violate it without losing an
optimum. For two stacks, the exact reduction to M2 and this structure
yield an $O(n^6)$-time dynamic program.

Extending this approach to three stacks requires selecting and
interleaving the orders of stacks~2 and~3. The exact reduction to M3
represents these interactions but leaves this optimization problem to be
solved. We develop a recursive decomposition that permits
polynomial-time optimization over a class of retrieval orders. Whether
this class contains an optimal order for every feasible three-stack
instance remains open. The computational study assesses the practical
benefit of exploiting this structure through computation times and energy
savings, and quantifies the additional savings from batching.

A preliminary version submitted to the GOR proceedings contains the
two-stack result, the three-stack reduction, and the realizability
theorem, and outlines the three-stack dynamic program
\citep{GORProceedings}. The present paper establishes within-stack
laminarity for every number of stacks and develops the full recursive
formulation and its guarantees.


Sections~\ref{sec:model}--\ref{sec:laminarity} formulate the problem and
establish the penalty identity and within-stack laminarity.
Section~\ref{sec:m2} solves the two-stack case;
Sections~\ref{sec:m3}--\ref{sec:dp-section} develop the three-stack
reduction, realizability characterization, and dynamic program.
Section~\ref{sec:study} evaluates computation times and energy savings
from sequencing and batching. Section~\ref{sec:conclusion} concludes
with open problems.

\section{Problem Statement and the Penalty Identity}\label{sec:model}

This section formulates the retrieval sequencing problem
(Section~\ref{sec:single}) and derives the penalty identity
(Section~\ref{sec:penalty}), which decomposes the energy of a retrieval
order into an instance-dependent constant and a pairwise penalty.

\subsection{The Storage System and Retrieval Problem}\label{sec:single}

We consider a compact storage system with side access, following the
physical model of \citet{SACRP-R1}. The system consists of rack modules,
called slices, each containing a row of adjacent vertical stacks. Since
slices operate independently, we consider one slice, represented by a
sequence of stacks \(T \coloneqq (1, \ldots, m)\), indexed from the entry
side. Stack \(t\)
initially contains \(n(t)\) uniform bins, placed directly on top of one
another without gaps. The height of a bin is the number of bins below it,
so stack \(t\) contains bins at heights
\(0, 1, \ldots, n(t)-1\).

Two lifting mechanisms provide access to the bins. A bot lift raises the
bot to the required height at the entry side. A top lift engages bins in
the stacks and raises each engaged bin together with all bins above it.
We assume sufficient vertical headroom for these lifts. Raising the
blocking bins in consecutive stacks creates a horizontal tunnel through
which the bot drives to reach a target bin from the side.

We are given a pick list \(\mathcal{B}\) of target bins,
where \(n \coloneqq |\mathcal{B}|\). We assume that the entire
pick list is available before retrieval begins.
Each target \(b \in \mathcal{B}\) has
stack \(s(b)\) and initial height \(h(b)\), with
\(0 \le h(b) < n(s(b))\), and distinct targets occupy distinct positions.
An instance is specified by the stack counts and target positions, with
all integers encoded in binary.

Retrieval proceeds in cycles, each of which retrieves exactly one target.
The bot drives on top of the bins left in place, without rails. To
retrieve target \(b\) from its current height \(h\), every preceding
stack \(t < s(b)\) must retain exactly \(h\) bins. If any preceding
stack contains fewer than \(h\) bins, \(b\) is not accessible because
the bot cannot cross that stack at height \(h\). Stack \(s(b)\) must
retain \(b\) and the \(h\) bins below it. Equivalently, the clearance
level is \(h\) in every preceding
stack and \(h+1\) in the target stack. No bins in later stacks are raised.
We call \(b\) accessible if these clearance levels are valid in the
current configuration. The resulting lifts clear the horizontal path to
\(b\) and remove the bins blocking its retrieval from above. After \(b\)
leaves the system, all lifted bins are lowered. The lifted substack above
\(b\) drops by one level, whereas bins lifted from preceding stacks
return to their previous heights. No non-target bin is relocated to another
stack. The energy expenditure of the cycle is the number of lifted bins.
Figure~\ref{fig:example} illustrates a two-cycle retrieval process.

A solution consists of exactly \(n\) cycles and is therefore determined
by a permutation \(\pi\) of \(\mathcal{B}\), the \emph{retrieval order}.
Feasibility and energy expenditure depend on this order, since earlier
retrievals change bin counts and target heights.
We write \(a \prec b\) if \(a\) is retrieved before \(b\) under \(\pi\),
and call \(\pi\) \emph{feasible} if every target is accessible when it is
retrieved. Its total energy \(E(\pi)\) is the sum of the energy
expenditures of its cycles. The objective is to find a feasible retrieval
order minimizing \(E(\pi)\).

For a target \(b\) in a retrieval order \(\pi\), let \(\kappa_t(b)\)
denote the number of targets retrieved from stack \(t\) before \(b\), and
let \(r_b\) denote the number of targets below \(b\) in its own stack
that are retrieved before \(b\), that is, the within-stack inversion count
of \(b\). When \(b\) is retrieved, stack \(t\)
contains \(n(t)-\kappa_t(b)\) bins, and \(b\) is at height
\(h(b)-r_b\). Hence, \(\pi\) is feasible if and only if
\begin{equation}\label{eq:retrieval-feasibility}
\kappa_t(b) \;\le\; n(t)-h(b)+r_b
\qquad
\text{for every } b\in\mathcal B \text{ and every } t<s(b).
\end{equation}

At the instance level, let \(R(b)\) denote the total number of targets
initially below \(b\) in its stack. By \citet{SACRP-R1}, an instance
admits a feasible retrieval order if and only if
\begin{equation}\label{eq:instance-feasibility}
n(t) \;\ge\; h(b)-R(b)
\qquad
\text{for every } b\in\mathcal B \text{ and every } t<s(b).
\end{equation}
Necessity follows because \(b\) can descend by at most \(R(b)\) levels.
Sufficiency follows by processing the stacks from right to left and each
stack from bottom to top. Throughout, \(h(b)\) denotes the initial height
of target \(b\), and \(n(t)\) denotes the initial number of bins in stack
\(t\).

\begin{figure}
\FIGURE
{\begin{tikzpicture}[
  x=0.54cm,
  y=0.50cm,
  bin/.style={draw=black!75,line width=0.35pt,fill=white,
              minimum width=0.54cm,minimum height=0.50cm,inner sep=0pt},
  target/.style={bin,fill=black!65,text=white,font=\scriptsize},
  lifted/.style={bin,fill=black!15},
  label/.style={font=\scriptsize}
]

\node[inner sep=0pt,minimum size=0pt] at (0,7.05) {};

\begin{scope}
  \draw[fill=white,line width=0.45pt] (-1.95,6.30) rectangle (4.15,6.58);
  \draw[line width=0.45pt] (-1.70,0) -- (-1.70,6.30);
  \draw[line width=0.45pt] (-1.95,0) -- (-1.95,6.30);
  \draw[line width=0.45pt] (-1.95,0) -- (4.15,0);

  \foreach \xx in {0,1.05,2.10,3.15} {
    \node[bin] at (\xx,0.5) {};
    \node[bin] at (\xx,1.5) {};
  }
  \node[target] at (0,2.5) {$b_1$};
  \node[bin] at (0,3.5) {};
  \node[bin] at (1.05,2.5) {};
  \node[bin] at (2.10,2.5) {};
  \node[bin] at (2.10,3.5) {};
  \node[target] at (3.15,2.5) {$b_2$};

  \draw[fill=black!28,line width=0.45pt,rounded corners=0.5pt]
    (-1.53,0.08) rectangle (-0.85,0.50);
  \fill (-1.39,0.03) circle (0.07);
  \fill (-1.00,0.03) circle (0.07);

  \foreach \xx/\ss in {0/1,1.05/2,2.10/3,3.15/4}
    \node[label] at (\xx,-0.34) {\(\ss\)};
  \node[label,font=\scriptsize\bfseries] at (1.58,-0.95) {(a)};
\end{scope}

\begin{scope}[xshift=3.80cm]
  \draw[fill=white,line width=0.45pt] (-1.95,6.30) rectangle (4.15,6.58);
  \draw[line width=0.45pt] (-1.70,0) -- (-1.70,6.30);
  \draw[line width=0.45pt] (-1.95,0) -- (-1.95,6.30);
  \draw[line width=0.45pt] (-1.95,0) -- (4.15,0);

  \foreach \xx in {0,1.05,2.10,3.15} {
    \node[bin] at (\xx,0.5) {};
    \node[bin] at (\xx,1.5) {};
  }
  \node[target] at (0,2.5) {$b_1$};
  \node[bin] at (1.05,2.5) {};
  \node[bin] at (2.10,2.5) {};
  \node[bin] at (2.10,3.5) {};
  \node[target] at (3.15,2.5) {$b_2$};

  \node[lifted] at (0,4.75) {};
  \draw[densely dashed,->,black!55,line width=0.45pt]
    (0,3.55) -- (0,4.17);

  \draw[fill=black!28,line width=0.45pt,rounded corners=0.5pt]
    (-1.28,2.13) rectangle (-0.72,2.55);
  \fill (-1.14,2.08) circle (0.07);
  \fill (-0.88,2.08) circle (0.07);
  \draw[densely dashed,->,line width=0.55pt]
    (-0.66,2.36) -- (-0.52,2.36);

  \foreach \xx/\ss in {0/1,1.05/2,2.10/3,3.15/4}
    \node[label] at (\xx,-0.34) {\(\ss\)};
  \node[label,font=\scriptsize\bfseries] at (1.58,-0.95) {(b)};
\end{scope}

\begin{scope}[xshift=7.85cm]
  \draw[fill=white,line width=0.45pt] (-1.95,6.30) rectangle (4.15,6.58);
  \draw[line width=0.45pt] (-1.70,0) -- (-1.70,6.30);
  \draw[line width=0.45pt] (-1.95,0) -- (-1.95,6.30);
  \draw[line width=0.45pt] (-1.95,0) -- (4.15,0);

  \foreach \xx in {0,1.05,2.10,3.15} {
    \node[bin] at (\xx,0.5) {};
    \node[bin] at (\xx,1.5) {};
  }
  \node[bin] at (0,2.5) {};
  \node[bin] at (1.05,2.5) {};
  \node[bin] at (2.10,2.5) {};
  \node[bin] at (2.10,3.5) {};
  \node[target] at (3.15,2.5) {$b_2$};

  \draw[fill=black!28,line width=0.45pt,rounded corners=0.5pt]
    (-1.53,0.08) rectangle (-0.85,0.50);
  \fill (-1.39,0.03) circle (0.07);
  \fill (-1.00,0.03) circle (0.07);

  \foreach \xx/\ss in {0/1,1.05/2,2.10/3,3.15/4}
    \node[label] at (\xx,-0.34) {\(\ss\)};
  \node[label,font=\scriptsize\bfseries] at (1.58,-0.95) {(c)};
\end{scope}

\begin{scope}[xshift=11.65cm]
  \draw[fill=white,line width=0.45pt] (-1.95,6.30) rectangle (4.15,6.58);
  \draw[line width=0.45pt] (-1.70,0) -- (-1.70,6.30);
  \draw[line width=0.45pt] (-1.95,0) -- (-1.95,6.30);
  \draw[line width=0.45pt] (-1.95,0) -- (4.15,0);

  \foreach \xx in {0,1.05,2.10,3.15} {
    \node[bin] at (\xx,0.5) {};
    \node[bin] at (\xx,1.5) {};
  }
  \node[target] at (3.15,2.5) {$b_2$};

  \node[lifted] at (0,4.75) {};
  \node[lifted] at (1.05,4.75) {};
  \node[lifted] at (2.10,4.75) {};
  \node[lifted] at (2.10,5.75) {};
  \foreach \xx in {0,1.05,2.10}
    \draw[densely dashed,->,black!55,line width=0.45pt]
      (\xx,3.55) -- (\xx,4.17);

  \draw[fill=black!28,line width=0.45pt,rounded corners=0.5pt]
    (-0.92,2.13) rectangle (-0.23,2.55);
  \fill (-0.78,2.08) circle (0.07);
  \fill (-0.37,2.08) circle (0.07);
  \draw[densely dashed,->,line width=0.55pt]
    (-0.12,2.36) -- (2.65,2.36);

  \foreach \xx/\ss in {0/1,1.05/2,2.10/3,3.15/4}
    \node[label] at (\xx,-0.34) {\(\ss\)};
  \node[label,font=\scriptsize\bfseries] at (1.58,-0.95) {(d)};
\end{scope}

\end{tikzpicture}}
{Illustration of two retrieval cycles. Consider a layout
with four stacks \(T\coloneqq(1,2,3,4)\) and pick list
\(\mathcal B\coloneqq\{b_1,b_2\}\), where
\(((s(b_1),h(b_1)),(s(b_2),h(b_2)))=((1,2),(4,2))\).
Figure~1(a) shows the initial configuration. In the first cycle
(Figure~1(b)), target \(b_1\) is retrieved with an energy expenditure of
\(1\), resulting in the configuration shown in Figure~1(c). In the second
cycle (Figure~1(d)), target \(b_2\) is retrieved with an energy expenditure
of \(4\). The resulting total energy expenditure is \(E(\pi)=1+4=5\).
Targets are highlighted in dark gray, and lifted bins in each cycle are
shown in light gray.
\label{fig:example}}
{}
\end{figure}

\subsection{The Penalty Identity}\label{sec:penalty}

Earlier retrievals affect the energy expenditure of the cycle retrieving
a target \(b\) in two ways. Namely, (i) a target retrieved before \(b\)
from \(b\)'s stack or a preceding stack removes one bin that the cycle of
\(b\) would otherwise lift, and (ii) a target retrieved before \(b\) from
below \(b\) in the same stack lowers \(b\) by one level, so that the cycle
lifts one additional bin in each of the \(s(b)\) stacks traversed by the
bot. These effects are additive across previously retrieved targets.
Summing them over all retrievals separates the total energy into an order-independent
constant and a pairwise penalty, as stated in the following lemma.

\begin{lemma}[Penalty identity]\label{lem:penalty}
For every feasible retrieval order \(\pi\),
\begin{equation}\label{eq:penalty-identity}
E(\pi) \;=\; K - \binom{n}{2} + P(\pi),
\qquad
K \coloneqq \sum_{b \in \mathcal{B}}
\Bigl[\sum_{t \le s(b)} n(t) - s(b)\,h(b) - 1\Bigr],
\end{equation}
where the \emph{baseline} \(K\) does not depend on \(\pi\) and the
\emph{penalty} \(P(\pi)\) is a sum over unordered pairs of targets. A
pair from distinct stacks contributes \(1\) if the target from the
higher-indexed stack is retrieved first, and a pair from a common stack
\(s\) contributes \(s\) if the lower target is retrieved first. Every
other pair contributes \(0\).

Equivalently, let \(m_t\) be the number of targets in stack \(t\). Using
\(r_b\) and \(\kappa_t(b)\) as defined above, we have
\begin{equation}\label{eq:penalty-element}
P(\pi)
\;=\;
\sum_{b \in \mathcal{B}}
\Bigl[
s(b)\,r_b
+
\sum_{t<s(b)} \bigl(m_t-\kappa_t(b)\bigr)
\Bigr].
\end{equation}
The first term counts targets below \(b\) that are retrieved before \(b\),
each with weight \(s(b)\); the second counts targets in preceding stacks
that are retrieved after \(b\).
\end{lemma}

\begin{proof}
Fix a target \(b\). Before \(b\) is retrieved, stack \(t\) contains
\(n(t)-\kappa_t(b)\) bins. Moreover, \(b\) descends one level whenever a
target below it in the same stack is retrieved. Thus, when \(b\) is
retrieved, it is at height \(h(b)-r_b\).

To clear the tunnel for \(b\), we lift every bin at height \(h(b)-r_b\)
or above in each preceding stack \(t<s(b)\), and every bin strictly above
\(b\) in stack \(s(b)\). Hence,
\[
\begin{aligned}
E_b
&= \sum_{t<s(b)}\left[(n(t)-\kappa_t(b))-(h(b)-r_b)\right]
 + \left[(n(s(b))-\kappa_{s(b)}(b))-(h(b)-r_b)-1\right] \\
&= \left[\sum_{t\leq s(b)}n(t)-s(b)h(b)-1\right]
 - \sum_{t\leq s(b)}\kappa_t(b) + s(b)r_b .
\end{aligned}
\]
Feasibility ensures that each preceding stack reaches height \(h(b)-r_b\)
and that stack \(s(b)\) contains \(b\) at that height. Summing over all
targets gives
\[
E(\pi) = K - A(\pi) + R(\pi),
\qquad
A(\pi) \coloneqq \sum_b \sum_{t \le s(b)} \kappa_t(b),
\qquad
R(\pi) \coloneqq \sum_b s(b)\, r_b.
\]

Both \(A\) and \(R\) can be written as sums over ordered pairs \((a,b)\)
with \(a \prec b\):
\[
A(\pi) = \#\{(a,b) : a \prec b,\ s(a) \le s(b)\},
\qquad
R(\pi) = \!\!\sum_{\substack{a \prec b,\ s(a) = s(b)\\ h(a) < h(b)}}\!\! s(b).
\]
Each unordered pair appears once, indexed by its earlier target.
\(R(\pi)\) counts every same-stack pair retrieved bottom-up, with weight
equal to the stack index. Thus,
\(R(\pi) = P_{\mathrm{same}}(\pi)\).

\(A(\pi)\) counts every same-stack pair once. For two targets in different
stacks, it counts the pair if the target in the lower-index stack is
retrieved first. Consequently, the only unordered pairs not counted by
\(A(\pi)\) are cross-stack pairs for which the target in the higher-index
stack is retrieved first. These are precisely the pairs contributing to
\(P_{\mathrm{cross}}(\pi)\). Hence,
\(A(\pi) = \binom{n}{2} - P_{\mathrm{cross}}(\pi)\). Combining this
identity with \(R(\pi)=P_{\mathrm{same}}(\pi)\) gives
\[
E(\pi) = K - A(\pi) + R(\pi)
= K - \binom{n}{2} + P(\pi).
\]

We index each same-stack bottom-up pair by its upper target. If \(a\) is
below \(b\) in the same stack and \(a \prec b\), then this pair is counted
by \(r_b\) and contributes \(s(b)\). Therefore,
\(P_{\mathrm{same}}(\pi) = \sum_b s(b)r_b\).
Moreover, we index each cross-stack pair by its earlier target \(b\), which lies in
the higher-index stack. For each \(t<s(b)\), there are
\(m_t-\kappa_t(b)\) targets in stack \(t\) retrieved after \(b\), each
forming one cross-stack pair with \(b\).
Summing over all targets and preceding stacks gives
\[
P_{\mathrm{cross}}(\pi) = \sum_{b} \sum_{t < s(b)} \bigl(m_t - \kappa_t(b)\bigr),
\qquad
P(\pi) = \sum_{b} \Bigl[\, s(b)\, r_b + \sum_{t < s(b)} \bigl(m_t - \kappa_t(b)\bigr) \Bigr]. \qedhere
\]
\end{proof}

Because \(K-\binom{n}{2}\) is independent of the retrieval order,
minimizing \(E(\pi)\) is equivalent to minimizing \(P(\pi)\). In the
per-element form, the feasibility condition~\eqref{eq:retrieval-feasibility}
becomes
\[
\kappa_t(b) \leq \min\bigl(m_t,n(t)-h(b)+r_b\bigr)
\qquad\text{for every }t<s(b).
\]

\section{Within-Stack Laminarity of Optima}\label{sec:laminarity}

In this section, we examine the order in which targets are retrieved from
the same stack. We show that, in every optimal retrieval order, the
lower targets retrieved before any given target form a consecutive block
in the stack's top-to-bottom sequence of targets, beginning with the next
target below it. We call an
order with this property within every stack a \emph{laminar-delay order}.
This result holds for any number of stacks. We prove it using an exchange
argument, called the \emph{stack-local promotion swap}, which transforms
any violation of this structure into a feasible order with strictly lower
energy expenditure.

Consider a stack \(s\) containing \(q_s\) targets, which we index from top
to bottom by \(1, \ldots, q_s\), so that their heights satisfy
\(h_1 > h_2 > \cdots > h_{q_s}\). For a feasible retrieval order \(\pi\)
and a target \(x\) in stack \(s\), we denote by \(B_x\) the set of targets
below \(x\) that are retrieved before \(x\):
\[
B_x \;\coloneqq\; \{\, y > x \text{ in stack } s :\ y \prec_\pi x \,\}
\;\subseteq\; \{x{+}1, \ldots, q_s\}.
\]
We refer to \(B_x\) as the \emph{within-stack inversion set} of \(x\).
Moreover, we call each target in \(B_x\) an \emph{inverted successor}
of \(x\).

We say that \(B_x\) is a \emph{prefix} if it consists of the first
\(|B_x|\) targets below \(x\) in top-to-bottom order, that is,
\(B_x=\{x{+}1,\ldots,x{+}|B_x|\}\).
We call an order of the sequence \(1,\ldots,q_s\) a
\emph{laminar-delay order} if the within-stack inversion set of every
target is a prefix.

For every preceding stack \(t<s\), we define the \emph{cap} of \(x\)
with respect to stack \(t\) as \(c_{x,t}\coloneqq n(t)-h_x\).
Each inverted successor retrieved before \(x\) lowers \(x\) by one level
and therefore increases its cap by one. Since \(r_x=|B_x|\), the
feasibility constraint~\eqref{eq:retrieval-feasibility} becomes
\[
\kappa_t(x)\leq \min\bigl(m_t,c_{x,t}+r_x\bigr)
\qquad\text{for every }t<s.
\]

Moreover, if \(x<y\), then \(x\) is above \(y\). Since consecutive
targets are at least one level apart, \(h_x\geq h_y+(y-x)\).
Consequently, \(c_{x,t}\leq c_{y,t}-(y-x)\) for every \(t<s\).

Suppose \(B_x\) is not a prefix. Let \(a\) be the largest index between
\(x\) and \(\max B_x\) that does not belong to \(B_x\). Thus, \(a\) is
retrieved after \(x\), whereas every target with index in
\((a,\max B_x]\) is retrieved before \(x\). Let \(b\) be the last of
these targets to be retrieved. Exchanging the positions of \(a\) and
\(b\) removes exactly one within-stack inversion and leaves all other cost
terms unchanged. The cap inequality above ensures that both targets remain
feasible after the exchange.

\begin{lemma}[Stack-local promotion swap]\label{lem:stackswap}
Let \(\pi\) be any feasible order, \(s\) a stack, and \(x\) a target of
stack \(s\) whose within-stack inversion set \(B_x\) is not a prefix. Let
\(M=\max B_x\), and let \(a\) be the largest index between \(x\) and
\(M\) that does not belong to \(B_x\). Then \(x<a<M\) and
\(\{a+1,\ldots,M\}\subseteq B_x\). Let \(b\) be the target in
\(\{a+1,\ldots,M\}\) that is retrieved last under \(\pi\). Exchanging
the positions of \(a\) and \(b\) in \(\pi\) yields a feasible order
\(\pi'\) with
\[
E(\pi') \;=\; E(\pi) - s,
\]
in which the inversion set of \(x\) becomes
\(B_x \setminus \{b\} \cup \{a\}\),
and the position, \(\kappa_t(w)\), and \(r_w\) of every target
\(w\notin\{a,b\}\) remain unchanged.
\end{lemma}


The proof of Lemma~\ref{lem:stackswap} is provided in
Appendix~\ref{app:laminarity}.

Based on Lemma~\ref{lem:stackswap}, the following result immediately
follows.

\begin{corollary}[Within-stack laminarity of optima]\label{cor:laminar}
In every optimal retrieval order, each target's within-stack inversion set
is a prefix. Equivalently, the restriction of the order to any stack is a
laminar-delay order of that stack's top-down sequence.
\end{corollary}

\begin{proof}
Suppose an optimal order contains a target \(x\) whose within-stack
inversion set is not a prefix. Lemma~\ref{lem:stackswap} then yields a
feasible order with strictly smaller energy, contradicting optimality. The
equivalence follows as within-stack inversion sets depend only on the
relative order of targets in the same stack.
\end{proof}

Corollary~\ref{cor:laminar} characterizes the within-stack structure of
an optimal order for any number of stacks. Thus, the search within each
stack is restricted to laminar-delay orders rather than arbitrary
permutations. Building on this result, we analyze the cases of two and
three stacks in Sections~\ref{sec:m2} and~\ref{sec:m3}, respectively.

\section{Two Stacks: An Exact Reduction and a Polynomial-Time Algorithm}\label{sec:m2}

We solve the two-stack problem in $O(n^6)$ time through an exact
reduction (Section~\ref{sec:m2-reduction}) and a dynamic program over
laminar-delay orders (Section~\ref{sec:m2-dp}). Optimality follows from
the top-target decomposition (Section~\ref{sec:m2-decomposition}) and a
feasible, cost-reducing promotion swap (Section~\ref{sec:m2-exchange}).

\subsection{An Exact Reduction to a One-Dimensional Scheduling Problem}
\label{sec:m2-reduction}

Let $m=2$. Observe that the targets in stack~1 can be retrieved from top to
bottom without loss of optimality. This order avoids every within-stack
penalty in stack~1, while stack~1 affects the feasibility and cost of a
stack~2 target only through the number of its targets already retrieved.
Consequently, we can represent stack~1 by this number. We
must therefore choose the order of the stack~2 targets and the number of
stack~1 targets retrieved before each of them. This yields the following
one-dimensional scheduling problem.

\begin{theorem}[Reduction to M2]\label{thm:reduction-m2}
Let an $m=2$ instance have $m_1$ targets in stack~1 and $q$ targets in
stack~2, the latter indexed from top to bottom by $j=1,\ldots,q$, with
caps $c_j\coloneqq n(1)-h_j$, which are strictly increasing integers.
Then minimizing $P$ is equivalent to the following problem
\textnormal{(M2)}: choose a permutation $\sigma$ of $[q]$ and, for each
target with index $j$, an integer $v_j\in[0,m_1]$ representing the number of stack~1
targets retrieved before $j$. The values $v_j$ are nondecreasing along
$\sigma$ and satisfy
\begin{equation}\label{eq:m2-feasibility}
v_j \;\le\; \min\bigl(m_1,\; c_j+r_j\bigr),
\qquad r_j\coloneqq\#\{\,j'>j:\; j'\prec_\sigma j\,\},
\end{equation}
Here, \(r_j\) is the within-stack inversion count of target \(j\). The objective is to
minimize
\begin{equation}\label{eq:m2-cost}
F(\sigma,v)
=
\sum_{j=1}^{q}(m_1-v_j)
+
2\sum_{j=1}^{q}r_j.
\end{equation}
The instance is feasible if and only if $c_j+(q-j)\geq 0$ for every
$j\in[q]$.
\end{theorem}

The proof is given in Appendix~\ref{app:m2}.

In the objective~\eqref{eq:m2-cost}, each stack~1 target that remains when a stack~2 target is
retrieved contributes one unit to the penalty. Each inversion among the
stack~2 targets contributes two units: retrieving a lower target first
lowers the upper target by one level, so that its retrieval requires lifting
one additional bin in each of the two traversed stacks
(Lemma~\ref{lem:penalty}). Retrieving $r_j$ lower stack~2 targets before
target $j$ also raises its cap from $c_j$ to $c_j+r_j$.

\subsection{Laminar-Delay Orders and the Dynamic Program}
\label{sec:m2-dp}

We next develop a polynomial-time algorithm for M2. Since $r_j\leq q-j$,
constraint~\eqref{eq:m2-feasibility} implies that every feasible solution satisfies
$v_j\leq\min\bigl(m_1,c_j+(q-j)\bigr)$.
Because the caps $c_j$ are strictly increasing integers, this upper bound is
nondecreasing in $j$ (see Appendix~\ref{app:m2},
Lemma~\ref{lem:deadlines}).

Nevertheless, it may be beneficial to retrieve a lower target before an
upper target. Suppose that target $j$ is delayed past one additional lower
target. Then $r_j$ increases by one, which adds $2$ to the objective. At the
same time, the upper bound on $v_j$ may increase by one. This may permit a
larger value of $v_j$, reducing the term $m_1-v_j$, and may also affect the
values assigned to other targets. Hence, whether target $j$ should be
delayed depends on the order and values of the remaining targets.

By Corollary~\ref{cor:laminar}, the order of the stack-2 targets in every
optimal solution of M2 is a laminar-delay order. The following recursive
characterization yields the decomposition used by
the dynamic program.

For a target \(j\), we call a set \(B\subseteq[j{+}1..q]\) an
\emph{index prefix} if \(B=\varnothing\) or
\(B=[j{+}1..j{+}r]\) for some \(1\leq r\leq q-j\).

\begin{lemma}[Characterization of laminar-delay orders]\label{lem:grammar}
An order \(\sigma\) of \([1..q]\) can be generated recursively by
\[
S([j..k])
\longrightarrow
j\cdot S([j{+}1..k])
\ \mid\
S([j{+}1..d])\cdot j\cdot S([d{+}1..k]),
\qquad j<d\leq k,
\]
where an empty interval generates the empty order, if and only if, for
every \(j\), its inversion set
\[
B_j\coloneqq\{j'>j: j'\prec_\sigma j\}
\]
is an index prefix. Thus, the generated orders are precisely the
laminar-delay orders. In this case, since \(r_j=|B_j|\), any two
intervals \([j..j{+}r_j]\) are either disjoint or one contains the other.
\end{lemma}

The proof is given in Appendix~\ref{app:m2}.

This recursive decomposition leads to the following dynamic program.

\begin{proposition}[Two-stack dynamic program]\label{thm:dp}
For $1\leq j\leq k+1\leq q+1$ and
$0\leq\lo\leq\hicmd\leq m_1$, let
$f(j,k,\lo,\hicmd)$ denote the minimum cost of a feasible laminar-delay
order on $[j..k]$ whose values lie in $[\lo,\hicmd]$. Set
$f(j,k,\cdot,\cdot)=0$ when $j>k$.

For a nonempty interval, target $j$ is either retrieved first with value
$v$, leaving the state $(j{+}1,k,v,\hicmd)$, or delayed past
$[j{+}1..d]$. In the latter case, its inversion count satisfies
$r_j=d-j$, and the
two remaining states are $(j{+}1,d,\lo,v)$ and
$(d{+}1,k,v,\hicmd)$. Hence,
\begin{equation}\label{eq:m2-dp}
\begin{aligned}
f(j,k,\lo,\hicmd)=\min\Bigl\{
&\min_{\substack{\lo\le v\le\hicmd\\ v\le c_j}}
   \bigl[(m_1-v)+f(j{+}1,k,v,\hicmd)\bigr],\\
&\min_{\substack{j<d\le k\\ \lo\le v\le\hicmd\\ v\le c_j+d-j}}
   \bigl[f(j{+}1,d,\lo,v)+(m_1-v)+2(d-j)+f(d{+}1,k,v,\hicmd)\bigr]
\Bigr\}.
\end{aligned}
\end{equation}

The dynamic program has $O(q^2m_1^2)$ states. For each state, the first
case contains $O(m_1)$ feasible values of $v$, while the second contains
$O(qm_1)$ feasible pairs $(d,v)$. Therefore, it runs in
$O(q^3m_1^3)\subseteq O(n^6)$ time.
\end{proposition}

\subsection{The Decomposition at the Top Target}
\label{sec:m2-decomposition}

In recurrence~\eqref{eq:m2-dp}, the set of targets retrieved
before $j$ is either empty or equal to $[j{+}1..d]$ for some
$j<d\leq k$. We now show that restricting the recurrence to these sets does
not exclude an optimal solution.

Consider an arbitrary order on an interval $I=[j..k]$ with window
$[\lo,\hicmd]$, and let $L=I\setminus\{j\}$. Let $B\subseteq L$ denote the
set of targets retrieved before $j$, and set $r=|B|$. The targets in $B$
are retrieved first, followed by $j$, while all targets in $L\setminus B$
are retrieved afterward. If $B$ is not a consecutive block, two additional
effects may arise. First, a target in $B$ and a target in $L\setminus B$ may
form an additional inversion. Second, retrieving the targets in $B$ may
increase the caps of some targets in $L\setminus B$. We formalize these two
effects as follows.

Because every target in $B$ is retrieved before every target in
$L\setminus B$, a pair $(b,s)\in B\times(L\setminus B)$ is an inversion
exactly when $b>s$. Let
$\cross(B)\coloneqq\#\{(b,s)\in B\times(L\setminus B):b>s\}$ denote the
number of these additional inversions. If $B=\{b_1<\dots<b_r\}$, then
$\cross(B)=\sum_{t=1}^{r}(b_t-j-t)$, since $b_t-j-t$ is the number of
targets in $L\setminus B$ with index smaller than $b_t$.

For a set $X$ of targets, consider the M2 instance obtained by deleting
every target outside $X$. The value $m_1$ and the caps of the targets in
$X$ remain unchanged, and the targets are relabeled consecutively in
increasing order of their original indices. For a window $W$, let
$\OPT(X,W)$ denote the minimum cost of this instance when all values are
restricted to $W$. Similarly, $\OPTp(X,W\mid B)$ denotes the minimum cost
under the same objective after increasing the cap of each $s\in X$ by
$\lvert\{b\in B:b>s\}\rvert$.

\begin{proposition}[Decomposition at the top target]\label{prop:decomp}
Let $B\subseteq L$, $r=|B|$, and
$v\in[\lo,\min(m_1,\hicmd,c_j+r)]$.

\emph{(a)} Suppose a feasible solution on $B$, with values in $[\lo,v]$,
has cost $C_B$. Suppose also that a feasible solution on $L\setminus B$,
with values in $[v,\hicmd]$, has cost $C_R$ after the cap of each
$s\in L\setminus B$ is increased by $\#\{b\in B:b>s\}$. Retrieving the
targets in $B$ according to the first solution, then $j$ with value $v$,
and finally the remaining targets according to the second solution gives
a feasible solution on $I$ with cost
$\Phi(B,v)=C_B+(m_1-v)+2r+C_R+2\,\cross(B)$.
Any solution on $L\setminus B$ that is feasible under the original caps
may also be used, since increasing the caps preserves feasibility.

\emph{(b)} Conversely, consider a feasible solution on $I$ in which
precisely the targets in $B$ are retrieved before $j$, and $j$ is assigned
value $v$. Its restriction to $B$ is feasible with values in $[\lo,v]$.
Its restriction to $L\setminus B$ uses values in $[v,\hicmd]$ and is
feasible under the increased caps. The total cost remains $\Phi(B,v)$.

If $B=P_r\coloneqq[j{+}1..j{+}r]$, then $b<s$ for every $b\in B$ and
$s\in L\setminus B$. Consequently, no cap is increased and
$\cross(B)=0$.
\end{proposition}

The proof is given in Appendix~\ref{app:m2}.

\subsection{The Core Exchange}
\label{sec:m2-exchange}

We now show that, in the decomposition of Proposition~\ref{prop:decomp},
an arbitrary set $B\subseteq L$ can be replaced by the index prefix $P_r$
of the same cardinality without increasing the cost. An adjacent exchange
may violate the cap of the upper target, whereas an arbitrary nonadjacent
exchange may also change the inversion counts of targets between the exchanged
positions. We therefore introduce the following promotion swap.

Consider a windowed M2 instance with targets $1,\dots,q$, strictly
increasing integer caps $c_1<\dots<c_q$, and a window $[\lo,\hicmd]$, where
$0\le\lo\le\hicmd\le m_1$. A feasible solution $(\sigma,v)$ has values that
are nondecreasing along $\sigma$ and satisfy
\[
\lo\le v_x\le\min(m_1,\hicmd,c_x+r_x)
\qquad\text{for every }x.
\]
Its cost is
$\sum_x(m_1-v_x)+2\sum_x r_x$.

Let $B_1\coloneqq\{x\in[2..q]:x\prec_\sigma 1\}$, and suppose that $B_1$ is
not an index prefix. Set $M\coloneqq\max B_1$ and
$a\coloneqq\max\{x<M:x\notin B_1\}$, and let $b$ be the target in $(a..M]$
retrieved last under $\sigma$. Then
$b\prec_\sigma 1\prec_\sigma a$. The promotion swap exchanges the
positions of $a$ and $b$, leaving the value assigned to each position
unchanged. By construction, every target $w$ between $b$ and $a$ in
$\sigma$ satisfies $w<a$ or $w>M$. Hence, its inversion count $r_w$ remains
unchanged. The strict increase of the cap sequence ensures that both
exchanged targets remain feasible. The swap replaces $B_1$ by
$(B_1\setminus\{b\})\cup\{a\}$, decreases the total inversion count by one,
and therefore decreases the cost by exactly $2$. The complete proof is
provided by Lemma~\ref{lem:promswap} in Appendix~\ref{app:core}.

Lemma~\ref{lem:promswap} implies that the inversion set of target~$1$ is an
index prefix in every optimal solution of a windowed instance (see
Appendix~\ref{app:core}, Corollary~\ref{cor:topprefix}). Applying this
property recursively to the two subproblems in Proposition~\ref{prop:decomp}
shows that the order of every optimal solution is a laminar-delay order
(see Appendix~\ref{app:core}, Corollary~\ref{cor:grammar}).

Accordingly, repeated promotion swaps transform a nonprefix inversion set
into the prefix of the same size, preserve the top value, and strictly
decrease the cost.

\begin{theorem}[Core exchange]\label{thm:core}
With the notation of Proposition~\ref{prop:decomp}, let $B\subseteq L$,
set $r=|B|$, and let $v\in[\lo,\min(m_1,\hicmd,c_j+r)]$.
Define
\[
\Phi(B,v)\coloneqq\OPT(B,[\lo,v])+(m_1-v)+2r
+\OPTp(L\setminus B,[v,\hicmd]\mid B)+2\,\cross(B).
\]
Then
\[
\Phi(P_r,v)\le\Phi(B,v).
\]
If $B\neq P_r$, then $\Phi(P_r,v)\le\Phi(B,v)-2$.
\end{theorem}

The proof is given in Appendix~\ref{app:core}.

Propositions~\ref{thm:dp} and~\ref{prop:decomp}, together with
Theorem~\ref{thm:core}, yield the following theorem.

\begin{theorem}[Polynomial-time solvability for two stacks]\label{thm:main}
For every window $[\lo,\hicmd]$ and every index interval $I=[j..k]$,
\[
\OPT(I,[\lo,\hicmd])=f(j,k,\lo,\hicmd).
\]
Consequently, the retrieval sequencing problem with two stacks can be
solved in $O(n^6)$ time.
\end{theorem}
The proof is given in Appendix~\ref{app:core}.

Thus, for two stacks, it is sufficient to consider retrieval orders in
which every delayed target is preceded by a consecutive block of targets
immediately below it. The dynamic program enumerates exactly these
recursive delay decisions and therefore computes an optimal retrieval
order in polynomial time.

\section{Three Stacks: Exact Reduction and Realizability}\label{sec:m3}

We extend the two-stack reduction to the merge problem M3, which
interleaves the retrieval orders of stacks~2 and~3
(Section~\ref{sec:m3reduction}), and characterize exactly which M3
instances arise from the storage model (Section~\ref{sec:realizability}).

\subsection{Reduction to the Merge Problem M3}\label{sec:m3reduction}

Consider an instance with $m=3$. As in the two-stack case, the targets in
stack~1 can be retrieved from top to bottom without loss of optimality.
This order avoids every within-stack penalty in stack~1, while stack~1
affects the feasibility and cost of targets in stacks~2 and~3 only through
the number of its targets already retrieved. Consequently,
we represent stack~1 by this number, denoted by $v=\kk_1$.

Let $S_2$ and $S_3$ denote the target sets of stacks~2 and~3,
respectively, and let $q_2\coloneqq |S_2|$ and
$q_3\coloneqq |S_3|$. Within each stack, we index the targets from top to
bottom, as in Section~\ref{sec:laminarity}. This yields the following
merge problem.

\begin{theorem}[Reduction to M3]\label{thm:reduction}
Consider an instance of the retrieval sequencing problem with $m=3$, and let $m_1$ denote the
number of targets in stack~1. For each $j\in S_2$, define
$\alpha_j\coloneqq n(1)-h_j$. For each $k\in S_3$, define
\[
\beta_k\coloneqq n(1)-h_k,
\qquad
\gamma_k\coloneqq n(2)-h_k=\beta_k+\goff,
\qquad
\goff\coloneqq n(2)-n(1).
\]
A solution to problem M3 consists of an order $\sigma$ of $S_2\cup S_3$
and, for each target $b$, an integer $v_b\in[0,m_1]$ representing the
number of stack-1 targets retrieved before $b$. The values $v_b$ are
nondecreasing along $\sigma$. Let
\[
r_j\coloneqq\#\{\,j'\in S_2:j'>j,\ j'\prec_\sigma j\,\},
\qquad
r_k\coloneqq\#\{\,k'\in S_3:k'>k,\ k'\prec_\sigma k\,\}.
\]
Here, $r_j$ and $r_k$ are the within-stack inversion counts of targets
$j$ and $k$, respectively.
For each $k\in S_3$, let
\[
\kk_2(k)\coloneqq\#\{\,j\in S_2:j\prec_\sigma k\,\}.
\]
The feasibility conditions are
\begin{equation}\label{eq:m3-feasibility}
\begin{aligned}
v_j&\le\min(m_1,\alpha_j+r_j) &&(j\in S_2),\\
v_k&\le\min(m_1,\beta_k+r_k) &&(k\in S_3),\\
\kk_2(k)&\le\min(q_2,\gamma_k+r_k) &&(k\in S_3),
\end{aligned}
\end{equation}
and the cost is
\begin{equation}\label{eq:m3-cost}
\sum_{j\in S_2}\bigl[2r_j+(m_1-v_j)\bigr]
+\sum_{k\in S_3}\bigl[3r_k+(m_1-v_k)+(q_2-\kk_2(k))\bigr].
\end{equation}
The minimum cost of problem M3 equals the minimum penalty $P(\pi)$ over
all feasible retrieval orders $\pi$.
\end{theorem}

Problem M3 differs from problem M2 through the interaction between
stacks~2 and~3. For each $k\in S_3$, retrieving more stack-2 targets before
$k$ reduces its contribution to~\eqref{eq:m3-cost} but may violate its
feasibility bound in~\eqref{eq:m3-feasibility}. The interleaving must
therefore balance these two effects.

Consider a feasible retrieval order in which the targets in stack~1 are
retrieved from top to bottom. Deleting these targets yields the order
$\sigma$, and setting $v_b$ equal to the number of stack-1 targets preceding
$b$ yields a feasible solution to problem M3. Conversely, any feasible
solution $(\sigma,v)$ to problem M3 defines a feasible retrieval order by
inserting the stack-1 targets into $\sigma$, in top-to-bottom order, so that
exactly $v_b$ of them precede each target $b$, and placing any remaining
stack-1 targets last. By Lemma~\ref{lem:penalty}, the two solutions have the
same cost. The complete proof is provided in
Appendix~\ref{app:m3reduction}.

\subsection{Realizability}\label{sec:realizability}

By Theorem~\ref{thm:reduction}, a three-stack instance is represented by a
tuple \((m_1,\alpha,\beta,\goff)\), where \(m_1\) is a nonnegative
integer, \(\goff\) is an integer, and \(\alpha\) and \(\beta\) are finite
integer sequences, possibly empty. We refer to this tuple as the
\emph{cap vector} of the instance. A cap vector is \emph{admissible} if both
\(\alpha\) and \(\beta\) are strictly increasing, and \emph{realizable}
if it is induced by an \(m=3\) instance.

Observe that every realizable cap vector is admissible, as target heights
strictly decrease with their top-to-bottom indices. If \(\alpha\) is
nonempty, the height range in stack~2 further implies
\(\alpha_1+\goff\ge 1\). Theorem~\ref{thm:real} shows that these necessary
conditions are also sufficient.

\begin{theorem}[Realizability]\label{thm:real}
Let \(m_1\) be a nonnegative integer, let \(\alpha\) and \(\beta\) be
finite integer sequences, possibly empty, and let \(\goff\) be an integer.
Then the tuple \((m_1,\alpha,\beta,\goff)\) is realizable if and only if it
is admissible and either \(\alpha\) is empty or
\(\alpha_1+\goff\ge 1\).
\end{theorem}

The proof is given in Appendix~\ref{app:m3reduction}.
Its sufficiency direction is constructive.

\section{The Joint-Interval Dynamic Program and its
Soundness}\label{sec:dp-section}

In this section, we develop a dynamic programming approach for the
three-stack problem M3. The main idea is based on a recursive decomposition
of the target sequences of stacks~2 and~3. At each step,
the uppermost remaining target of either stack is retrieved immediately
or delayed behind a structured block of targets. We characterize the
class of retrieval orders generated by this decomposition and show that
a minimum-cost order in this class can be computed in polynomial time.

Consider a subproblem defined by one interval of targets from each of
stacks~2 and~3. The dynamic program selects the uppermost target of
either interval. It may retrieve this target first within the
subproblem. Alternatively, it may delay the target until after a block
containing a nonempty consecutive sequence of the next targets in the
same stack and, possibly, a consecutive sequence beginning with the
uppermost target of the other stack. The targets from the two sequences
may be interleaved within the block. The same decomposition is then
applied recursively to this block and to the targets retrieved after
the delayed target.

The solution space explored by the dynamic program consists of all
retrieval orders that can be constructed entirely through these two
cases. Figure~\ref{fig:delayedprod} illustrates the delayed case. A
formal definition of the recursive decomposition is given in
Definition~\ref{def:productions} in Appendix~\ref{app:dp}.

\begin{figure}
\FIGURE
{
%
\begin{tikzpicture}[
  x=1cm,
  y=0.6cm,
  elem/.style={draw=black!75,line width=0.35pt,fill=white,
               minimum width=1.15cm,minimum height=0.62cm,inner sep=0pt,
               font=\scriptsize},
  blockelem/.style={elem,fill=black!10},
  delayedelem/.style={elem,draw=black,line width=0.9pt,fill=black!22},
  lab/.style={font=\scriptsize},
  brc/.style={decorate,decoration={brace,amplitude=3pt},line width=0.45pt}
]
\node[blockelem]   (e1) at (0,0)    {$(3,1)$};
\node[blockelem]   (e2) at (1.3,0)  {$(2,2)$};
\node[delayedelem] (e3) at (2.6,0)  {$(2,1)$};
\node[elem]        (e4) at (3.9,0)  {$(2,3)$};
\node[elem]        (e5) at (5.2,0)  {$(3,2)$};

\draw[->,line width=0.45pt] (-0.7,-2.55) -- (5.9,-2.55);
\node[lab,anchor=north] at (2.6,-2.65) {retrieval order};

\draw[brc] (-0.55,0.48) -- (1.85,0.48);
\node[lab,align=center,anchor=south] at (0.65,0.72)
  {block: next stack-2 target $(2,2)$,\\ uppermost stack-3 target $(3,1)$};
\draw[brc] (2.05,0.48) -- (3.15,0.48);
\node[lab,align=center,anchor=south] (delab) at (2.6,2.05)
  {delayed target $(s,j_s)=(2,1)$};
\draw[line width=0.35pt,black!55] (2.6,2.0) -- (2.6,0.85);
\draw[brc] (3.35,0.48) -- (5.75,0.48);
\node[lab,anchor=south] at (4.55,0.72) {suffix};

\draw[brc] (1.85,-0.48) -- (-0.55,-0.48);
\node[lab,anchor=north] at (0.65,-0.72) {values in $[\lo,v]$};
\node[lab,anchor=north] at (2.6,-0.72) {$\{v\}$};
\draw[brc] (5.75,-0.48) -- (3.35,-0.48);
\node[lab,anchor=north] at (4.55,-0.72) {values in $[v,\hicmd]$};
\end{tikzpicture}}
{Illustration of the delayed case. Target \((2,1)\) is retrieved after a
block containing targets \((2,2)\) and \((3,1)\), followed by the remaining
targets. Values lie in \([\lo,v]\) within the block and in
\([v,\hicmd]\) afterward; the delayed target receives value \(v\).
Each target is labeled by its stack and top-to-bottom index.
\label{fig:delayedprod}}
{}
\end{figure}

\begin{proposition}[Three-stack dynamic program]\label{def:dp}
A state comprises two index intervals $[j_2..k_2]$ and $[j_3..k_3]$, a
value window $[\lo,\hicmd]$, and an integer $B$. The intervals identify
the stack-2 and stack-3 targets in the subproblem. The value window bounds
the number of stack-1 targets that may precede these targets, while $B$
equals the number of stack-2 targets retrieved before any target in the
subproblem.

Let $f(j_2,k_2,j_3,k_3,\lo,\hicmd,B)$ denote the minimum cost for this
state over feasible recursively decomposable orders, with $f=0$ when
both intervals are empty. For a nonempty state, cases
\textnormal{(I2)} and \textnormal{(I3)} immediately retrieve the uppermost remaining
target from stacks~2 and~3, respectively, while cases \textnormal{(D2)}
and \textnormal{(D3)} delay the corresponding target until after a block
of targets has been retrieved. In a delayed case, the block is assigned
values in $[\lo,v]$, the delayed target receives value $v$, and the
remaining targets are assigned values in $[v,\hicmd]$. The resulting
recurrence is
\begin{equation}\label{eq:m3-dp}
f(j_2,k_2,j_3,k_3,\lo,\hicmd,B)=\min\bigl\{\text{(I2)},\ \text{(I3)},\ \text{(D2)},\ \text{(D3)}\bigr\},
\end{equation}
where
\begin{align*}
\text{(I2)}\ &\min_{v}\ \bigl[(m_1{-}v)+f(j_2{+}1,k_2,j_3,k_3,v,\hicmd,B{+}1)\bigr],\\
  &\qquad v\in[\lo,\min(\hicmd,m_1,\alpha_{j_2})];\\[2pt]
\text{(I3)}\ &\min_{v}\ \bigl[(m_1{-}v)+(q_2{-}B)+f(j_2,k_2,j_3{+}1,k_3,v,\hicmd,B)\bigr],\\
  &\qquad v\in[\lo,\min(\hicmd,m_1,\beta_{j_3})],\ B\le\min(q_2,\gamma_{j_3});\\[2pt]
\text{(D2)}\ &\min_{d_2,d_3,v}\ \bigl[f(j_2{+}1,d_2,j_3,d_3,\lo,v,B)+2(d_2{-}j_2)+(m_1{-}v)\\
  &\qquad{}+f(d_2{+}1,k_2,d_3{+}1,k_3,v,\hicmd,B{+}(d_2{-}j_2{+}1))\bigr],\\
  &\qquad v\in[\lo,\min(\hicmd,m_1,\alpha_{j_2}{+}(d_2{-}j_2))];\\[2pt]
\text{(D3)}\ &\min_{d_2,d_3,v}\ \bigl[f(j_2,d_2,j_3{+}1,d_3,\lo,v,B)+3(d_3{-}j_3)+(m_1{-}v)+(q_2{-}K_2)\\
  &\qquad{}+f(d_2{+}1,k_2,d_3{+}1,k_3,v,\hicmd,K_2)\bigr],\\
  &\qquad v\in[\lo,\min(\hicmd,m_1,\beta_{j_3}{+}(d_3{-}j_3))],\ K_2\le\min(q_2,\gamma_{j_3}{+}(d_3{-}j_3)).
\end{align*}
Here $K_2=B+(d_2{-}j_2{+}1)$. Cases \textnormal{(I2)} and
\textnormal{(D2)} apply only when $j_2\le k_2$, while cases
\textnormal{(I3)} and \textnormal{(D3)} apply only when $j_3\le k_3$.
In \textnormal{(D2)}, the block boundaries satisfy
$j_2<d_2\le k_2$ and $j_3{-}1\le d_3\le k_3$. In \textnormal{(D3)},
they satisfy $j_2{-}1\le d_2\le k_2$ and $j_3<d_3\le k_3$.
These bounds ensure that the block contains at least one additional
target from the same stack as the delayed target. If no admissible value
of $v$ exists, the corresponding case has value $+\infty$.
In the delayed cases, the selected target remains unretrieved while the
block is processed. The block subproblem therefore retains the same
value of $B$, while the remaining subproblem uses the updated value
$K_2$.
\end{proposition}

\begin{proposition}[State count and running time]\label{prop:dp-time}
The joint-interval DP has $O(q_2^3 q_3^2 m_1^2)$ states. The value of its
initial state, $f(1,q_2,1,q_3,0,m_1,0)$,
can be computed in $O(q_2^4 q_3^3 m_1^3)\subseteq O(n^{10})$ time by
evaluating the states in order of increasing interval size.
\end{proposition}

\begin{proof}
The interval of stack-2 targets, interval of stack-3 targets, value window, and parameter $B$
have $O(q_2^2)$, $O(q_3^2)$, $O(m_1^2)$, and $O(q_2)$ possible choices,
respectively, giving $O(q_2^3q_3^2m_1^2)$ states. Each case in
recurrence~\eqref{eq:m3-dp} refers only
to states with fewer remaining targets, so the states can be evaluated
in increasing interval size. Cases \textnormal{(I2)} and
\textnormal{(I3)} examine $O(m_1)$ values of $v$, whereas cases
\textnormal{(D2)} and \textnormal{(D3)} examine $O(q_2q_3m_1)$ triples
$(d_2,d_3,v)$. The total running time is therefore
\[
O(q_2^3q_3^2m_1^2)\cdot O(q_2q_3m_1)
=
O(q_2^4q_3^3m_1^3)
\subseteq O(n^{10}).
\]
\end{proof}

\begin{proposition}[Soundness of the dynamic program]\label{prop:assembly}
\leavevmode\newline
If the initial-state value $f(1,q_2,1,q_3,0,m_1,0)$ is finite, then there
exists a feasible schedule of that cost. Hence,
$\mathrm{DP}\ge\OPT$.
\end{proposition}

The proof is given in Appendix~\ref{app:dp}.

The parameter $B$ cannot be inferred from the left endpoint $j_2$ of
the interval of stack-2 targets. When the uppermost stack-2 target is delayed, it
is excluded from the block subproblem but remains unretrieved until the
block is completed. Consequently, setting $B=j_2-1$ overcounts the
number of stack-2 targets retrieved before the block.
Example~\ref{ex:bstate} shows that this substitution can cause the
dynamic program to return a value below the optimum.

\begin{example}[failure of the substitution $B=j_2-1$]\label{ex:bstate}
Consider the instance with parameters
$m_1=4,\qquad \alpha=(2,3,5),\qquad \beta=(-1,0),\qquad \goff=4$.
The optimum is $20$, whereas the dynamic program obtained by setting
$B=j_2-1$ in every state returns $19$.
\end{example}

\begin{proposition}[Minimum over the explored solution space]\label{prop:parsing-min}
For every instance of problem M3, the dynamic program returns the
minimum cost among all feasible solutions whose retrieval orders it
explores. If no such solution exists, it returns $+\infty$.
\end{proposition}

The proof is given in Appendix~\ref{app:dp}.

Proposition~\ref{prop:parsing-min} establishes that the dynamic program
computes a minimum-cost solution over the recursively decomposable orders.
We conjecture that this class contains an optimal solution for every
feasible realizable instance of problem M3.

\begin{conjecture}[Completeness of the recursive decomposition]\label{lem:aweak}
For every feasible order \(\pi\) of a realizable instance of problem M3, there
exists a recursively decomposable feasible order whose cost
does not exceed the cost of \(\pi\).
\end{conjecture}

Across verification test sets with $q_2+q_3\leq14$, no feasible order was
found with cost below the recursive-decomposition minimum. The dynamic
program also agreed with exhaustive enumeration on 3,268 instances with
$q_2+q_3\leq13$, with zero mismatches.

\section{Computational Study}\label{sec:study}

We describe the experimental setup (Section~\ref{sec:study-setup}),
compare running times (Section~\ref{sec:study-runtime}), and measure
energy savings over greedy retrieval (Section~\ref{sec:study-savings})
and from batching (Section~\ref{sec:study-price}).

\subsection{Experimental Setup}\label{sec:study-setup}

As operational data are not available, our experiments use synthetic
instances. The storage layouts and target positions are generated
randomly, following the procedure of \citet{SACRP-R1}. We consider two
and three stacks with $n\in\{6,12,18,24,30,36,48,60\}$ targets.
For each stack containing at least one target, we calculate the proportion
of target bins. The instance's target density is the average of these
proportions. We distinguish between three target densities, namely,
low (below $0.50$), medium ($0.50$ to $0.70$), and high (above $0.70$).
We retain a generated
instance only if it is feasible. Moreover, every stack contains at
least one bin, and the instance's target density falls within the
specified range. We generate 30 instances for each of the 48 parameter
combinations, yielding 1,440 instances in total.

All algorithms are implemented in C\# and run under .NET~8 on a Windows
workstation with an AMD Ryzen Threadripper PRO 9965WX 24-core processor. Each algorithm runs in
a separate process restricted to a single core, with a time limit of
600 seconds and a memory limit of 12\,GiB.

\subsection{Computational Performance}\label{sec:study-runtime}\label{sec:study-verify}

We first examine the computational benefits of exploiting the structural
properties established in this paper. We compare the two-stack DP of
Section~\ref{sec:m2} and the three-stack DP of
Section~\ref{sec:dp-section} with an exact exponential-time dynamic
program whose states represent subsets of retrieved targets, referred
to as the \emph{subset DP}. Table~\ref{tab:completion} reports the number
of instances solved within the computational limits and the average
running time over all 90 instances, with timeouts recorded at 600
seconds. Figure~\ref{fig:runtime} illustrates how average running times
increase with the number of targets.

\begin{table}[htbp]
\TABLE
{Computational performance by number of stacks and targets. Each row
contains 90 instances across the three target densities. ``Completed''
gives the number of runs completed within the computational limits.
The time limit is set at 600 seconds.\label{tab:completion}}
{\begin{tabular}{r rr rr rr}
\toprule
 & \multicolumn{2}{c}{Two-/three-stack DP} & \multicolumn{2}{c}{Subset DP}
 & \multicolumn{2}{c}{Batched DP}\\
\cmidrule(lr){2-3}\cmidrule(lr){4-5}\cmidrule(lr){6-7}
$n$ & Completed & \shortstack{Average\\time [s]}
 & Completed & \shortstack{Average\\time [s]}
 & Completed & \shortstack{Average\\time [s]}\\
\midrule
\multicolumn{7}{l}{\emph{Two stacks}}\\
 6 & 90 & 0.004 & 90 &   0.001 & 90 &   0.001\\
12 & 90 & 0.004 & 90 &   0.009 & 90 &   0.028\\
18 & 90 & 0.004 & 90 &   0.396 & 90 &   1.507\\
24 & 90 & 0.006 & 90 &  16.116 & 77 & 312.921\\
30 & 90 & 0.008 &  0 & 600.000 &  0 & 600.000\\
36 & 90 & 0.017 &  0 & 600.000 &  0 & 600.000\\
48 & 90 & 0.060 &  0 & 600.000 &  0 & 600.000\\
60 & 90 & 0.218 &  0 & 600.000 &  0 & 600.000\\
\addlinespace
\multicolumn{7}{l}{\emph{Three stacks}}\\
 6 & 90 &   0.002 & 90 &   0.001 & 90 &   0.001\\
12 & 90 &   0.003 & 90 &   0.010 & 90 &   0.031\\
18 & 90 &   0.021 & 90 &   0.420 & 90 &   2.186\\
24 & 90 &   0.273 & 90 &  17.971 & 55 & 425.641\\
30 & 90 &   1.068 &  0 & 600.000 &  0 & 600.000\\
36 & 90 &   5.609 &  0 & 600.000 &  0 & 600.000\\
48 & 90 & 114.380 &  0 & 600.000 &  0 & 600.000\\
60 & 22 & 532.928 &  0 & 600.000 &  0 & 600.000\\
\bottomrule
\end{tabular}}
{}
\end{table}

\begin{figure}[htbp]
\FIGURE
{\includegraphics[width=\linewidth]{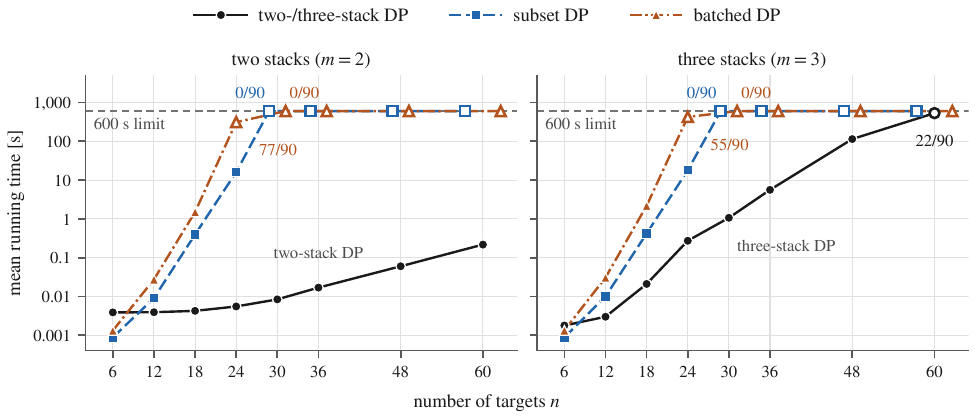}%
}
{Average running times by number of targets for two stacks (left) and
three stacks (right), on a logarithmic time axis. Each point averages
90 instances, with runs that did not complete within the computational
limits counted at 600 seconds (dashed line). Filled markers indicate
that all runs completed; hollow markers indicate otherwise. Annotations
give the number of completed runs out of 90.\label{fig:runtime}}
{}
\end{figure}

Both dynamic programs found optimal solutions for all instances with at
most 24 targets, matching the subset DP on all 360 instances for each
stack count. For the two-stack DP, this confirms the correctness of the
implementation. For the three-stack DP, these results provide further support for
Conjecture~\ref{lem:aweak}.

The two-stack DP solved all instances, with average runtimes increasing
from 0.004 seconds for 6 targets to 0.218 seconds for 60 targets. The
three-stack DP completed all runs with up to 48 targets, but its
runtimes increased more sharply, averaging 114.4 seconds at that size.
For 60 targets, it completed 22 of the 90 runs, with an average runtime
of 532.9 seconds including timeouts. In comparison, the subset DP
required an average of 16.1 and 18.0 seconds for 24 targets with two and
three stacks, respectively, and reached the time limit on all larger
instances.

\subsection{Energy Savings Relative to Greedy Retrieval}\label{sec:study-savings}

In this section, we compare optimized retrieval orders with a
greedy heuristic based on the approach of \citet{SACRP-R1}, modified to
retrieve exactly one target per cycle. At each step, the
heuristic retrieves the highest accessible target whose removal
preserves feasibility of the remaining instance, breaking ties by
choosing the target in the rightmost stack. The pseudocode of the
greedy heuristic is given in Algorithm~\ref{alg:greedy}.

\begingroup
\setlength{\intextsep}{0pt}
\begin{algorithm}[H]
\caption{Greedy heuristic}\label{alg:greedy}
\begingroup
\SingleSpacedXII
\begin{algorithmic}[1]
\Require A feasible instance of the retrieval sequencing problem
\State Initialize the retrieved set $X\gets\emptyset$ and retrieval order $\pi\gets()$
\While{$X\neq\mathcal B$}
    \State Let $A$\textsuperscript{*} contain every accessible $b\in\mathcal B\setminus X$ whose retrieval preserves feasibility of the remaining instance
    \State Select from $A$ a target $b$ of maximum current height, breaking ties by maximum $s(b)$
    \State Retrieve $b$, set $X\gets X\cup\{b\}$, and append $b$ to $\pi$
\EndWhile
\State \Return $\pi$
\end{algorithmic}
\endgroup
\end{algorithm}
\endgroup
\begingroup
\SingleSpacedXII
\noindent{\footnotesize \textsuperscript{*}The set $A$ is nonempty at every iteration because the remaining instance is feasible. Hence, the heuristic returns a feasible retrieval order.}\par
\endgroup

For each instance, we calculate the energy saving as
$100(E_{\mathrm G}-E_{\mathrm S})/E_{\mathrm G}$, where $E_{\mathrm G}$
is the energy expenditure of the greedy heuristic and $E_{\mathrm S}$
is the minimum energy expenditure for single-target retrieval. The
penalties returned by the two-stack and three-stack DPs are converted
to energy expenditure using Lemma~\ref{lem:penalty}. We consider the 720 instances
with at most 24 targets, for which the subset DP establishes the optimum.
On all these instances, $E_{\mathrm G}>0$.
Table~\ref{tab:energy-reductions} reports the average and median energy
savings by number of stacks and target density (see the columns under
``Saving over greedy'').

\begin{table}[htbp]
\TABLE
{Energy savings (\%) by number of stacks and target density.
``Saving over greedy'' compares optimal single-target retrieval with
the greedy heuristic ($n\leq24$). ``Reduction from batching'' compares
optimal batched retrieval with optimal single-target retrieval
($n\leq18$). $N$ denotes the number of
instances included in each average and median.\label{tab:energy-reductions}}
{\setlength{\tabcolsep}{8pt}%
\begin{tabular}{ll rrr@{\qquad}rrr}
\toprule
 & & \multicolumn{3}{c}{Saving over greedy}
 & \multicolumn{3}{c}{Reduction from batching\textsuperscript{*}}\\
 & & \multicolumn{3}{c}{(through 24 targets)}
 & \multicolumn{3}{c}{(through 18 targets)}\\
\cmidrule(lr){3-5}\cmidrule(lr){6-8}
Stacks & Target density & $N$ & Average & Median & $N$ & Average & Median\\
\midrule
2 & Low    & 120 & 13.61 & 13.57 &  90 & 28.77 & 31.83\\
2 & Medium & 120 & 22.68 & 22.46 &  90 & 35.86 & 40.00\\
2 & High   & 120 & 46.81 & 43.17 &  82 & 48.71 & 53.39\\
2 & All    & 360 & 27.70 & 22.60 & 262 & 37.45 & 36.51\\
\addlinespace
3 & Low    & 120 & 18.30 & 17.79 &  90 & 24.34 & 24.18\\
3 & Medium & 120 & 32.49 & 30.46 &  89 & 38.48 & 40.00\\
3 & High   & 120 & 59.11 & 55.56 &  73 & 46.44 & 50.00\\
3 & All    & 360 & 36.63 & 30.30 & 252 & 35.74 & 33.84\\
\addlinespace
\multicolumn{2}{l}{Overall} & 720 & 32.17 & 26.17 & 514 & 36.61 & 35.71\\
\bottomrule
\end{tabular}}
{\footnotesize \textsuperscript{*}Of the 90 instances per combination of
stack count and target density, those with zero optimal single-target
energy expenditure are excluded from the percentages: 8 (two stacks,
high), 1 (three stacks, medium), and 17 (three stacks, high); 26 in total.}
\end{table}

Optimizing the retrieval order reduces energy expenditure by $32.17\%$
on average compared with the greedy heuristic, with a median saving of
$26.17\%$. The optimal solution improves on the greedy solution in 719
of the 720 instances and has the same energy expenditure in one. The
average saving increases with target density from $13.61\%$ to
$46.81\%$ for two stacks, and from $18.30\%$ to $59.11\%$ for three
stacks. At every target density, the average saving is higher for three
stacks. Thus, optimizing the retrieval order is particularly beneficial
on the tested instances with high target density.

For two stacks, the average energy saving is $25.37\%$ and the median
is $21.83\%$ across all 720 two-stack instances. For three stacks with
30 to 60 targets, the DP completed 292 runs. In every completed run,
the resulting retrieval order consumed less energy than the greedy
solution, with an average saving of $30.51\%$ and a median of
$28.02\%$.

\subsection{Energy Benefit of Batching}\label{sec:study-price}

In this section, we allow several targets to be retrieved in one cycle,
as in the batched retrieval formulation of \citet{SACRP-R1}. We examine
the additional energy savings from batching and the computation time
required to obtain them. For batched retrieval, we use their original
dynamic programming implementation, including its dominance rules.
Note that every
single-target retrieval order is also a feasible batched solution, with
one target retrieved per cycle. Hence, the minimum energy expenditure
under batched retrieval cannot exceed that under single-target retrieval.

We consider the 540 instances with at most 18 targets, for which both
the subset DP and the batched DP found optimal solutions. For each
instance, we calculate the additional energy saving from batching as
$100(E_{\mathrm S}-E_{\mathrm B})/E_{\mathrm S}$, where $E_{\mathrm B}$
denotes the minimum energy expenditure under batched retrieval. Both
optimal energy expenditures are zero on 26 instances. We exclude these
instances from the percentages, leaving 514 instances.
Table~\ref{tab:energy-reductions} reports the average and median savings
by number of stacks and target density (see the columns under
``Reduction from batching'').

Batching reduces energy expenditure by $36.61\%$ on average, with a
median saving of $35.71\%$. The optimal batched solution improves on
the optimal single-target solution in 474 of the 540 instances and has
the same energy expenditure in 66. The average saving increases with
target density from $28.77\%$ to $48.71\%$ for two stacks, and from
$24.34\%$ to $46.44\%$ for three stacks. Overall, the average savings
are similar for both stack counts, at $37.45\%$ and $35.74\%$,
respectively.

At 24 targets, both single-target DPs completed all instances in under
0.3 seconds on average, whereas the batched DP completed only 77 of 90
two-stack instances and 55 of 90 three-stack instances within the
computational limits (Table~\ref{tab:completion}). The batched DP
completed no run with 30 or more targets.

\section{Conclusion}\label{sec:conclusion}
\label{sec:extensions}

Developing efficient retrieval policies for a recently introduced storage
system requires understanding which decisions can be optimized efficiently
and which interactions make the problem difficult. The single-target case
provides a starting point for this analysis by separating retrieval
sequencing from the grouping of targets into cycles. Our results show
that the accessibility and energy effects of side access permit a
structural restriction: every optimal retrieval order is laminar within
each stack, regardless of the number of stacks. Together with the
pairwise representation of energy expenditure, this structure yields a
polynomial-time algorithm for two stacks. For three stacks, the recursive
decomposition provides a way to coordinate retrievals across stacks and
optimize over the resulting class of orders in polynomial time. The
within-stack theorem therefore supplies a basis for algorithm design that
extends beyond the two-stack case.

The computational results show the practical value of exploiting this
structure. On all tested instances with at most 24 targets, both dynamic
programs matched the exact subset dynamic program, and optimizing the
retrieval order reduced energy expenditure by 32.17\% on average relative
to the greedy heuristic. At 24 targets, both algorithms required less
than 0.3 seconds on average. These results suggest that accounting for the
effects of retrieval decisions can yield substantial energy savings
within short computation times. Batching offers further savings,
averaging 36.61\% on the instances with at most 18 targets and positive
optimal single-target energy expenditure, but requires substantially more
computation time. Thus, the single-target case also provides a reference
for assessing the additional energy benefit and computational effort
associated with batching.

Whether the recursive decomposition contains an optimal retrieval order
for every feasible three-stack instance remains open. Establishing this
would extend polynomial-time solvability from two to three stacks. For
larger numbers of stacks, the within-stack structure established here
continues to restrict optimal retrieval orders, but exploiting it
algorithmically requires further understanding of the interactions
between stacks. This raises the question of whether retrieval sequencing
is polynomially solvable for every fixed number of stacks, and whether it
becomes NP-hard when the number of stacks is part of the input. Resolving
these questions would clarify how far the computational advantage of the
single-target case extends. The computational results also motivate
studying retrieval policies that limit the number of targets per cycle
and connect the single-target case with batched retrieval. Studying how
this limit affects energy expenditure and computation time would help
determine how much of the energy benefit of batching can be retained
while allowing faster optimization.

\bibliographystyle{informs2014}
\bibliography{refs}

\clearpage

\begin{APPENDICES}

\section*{Appendix}\label{app:proofs}
\renewcommand{\thesubsection}{A.\arabic{subsection}}

\subsection{Stack-Local Promotion for Every $m$}\label{app:laminarity}

We use the notation of Section~\ref{sec:laminarity}. If $x<y$ are targets
in stack $s$, then $x$ is at least $y-x$ levels above $y$. Consequently,
for every preceding stack $t<s$,
\begin{equation}\label{eq:gradient}
c_{x,t}\le c_{y,t}-(y-x).
\end{equation}

\begin{proof}[Proof of Lemma~{\normalfont\ref{lem:stackswap}}]
Since $b\in B_x$ and $a\notin B_x$, with $a>x$, we have
$b\prec_\pi x\prec_\pi a$. Moreover, since $b$ is retrieved last among the
targets in $\{a+1,\ldots,M\}$, every other target in this set is retrieved
before $b$.

\emph{Step 1 (Targets between the exchanged positions).} Let
$w\notin\{a,b\}$ be a target in stack $s$ that is retrieved strictly between
$b$ and $a$. We show that $w<a$ or $w>M$.

If $w=x$, the claim follows from $x<a$. Suppose first that $w\prec_\pi x$.
If $w>a$, then $w\in B_x\cap\{a+1,\ldots,M\}$. Since $b\prec_\pi w$, this
contradicts the choice of $b$ as the last target in $\{a+1,\ldots,M\}$
retrieved under $\pi$. Hence, $w<a$.

Now suppose that $x\prec_\pi w$. If $w<x$, then $w<a$. Otherwise, $w>x$ and
$w\notin B_x$. Because every target in $\{a+1,\ldots,M\}$ belongs to $B_x$,
and $w\neq a$, it follows that $w<a$ or $w>M$.

\emph{Step 2 (Effects on all other targets).} Consider any target
$w\notin\{a,b\}$. If $w$ belongs to a stack other than $s$, the exchange does
not change the relative order of the targets in its stack. Therefore, the
number $r_w$ of lower targets retrieved before $w$ remains unchanged.

Now suppose that $w$ belongs to stack $s$. The set of targets from stack $s$
retrieved before $w$ can change only if $w$ is retrieved strictly between
$b$ and $a$. In that case, $b$ is replaced by $a$ in this set. By Step~1,
either $w<a<b$ or $w>M\ge b$. In the first case, both $a$ and $b$ are below
$w$; in the second case, both are above $w$. Hence, replacing $b$ by $a$ does
not change the number of lower targets retrieved before $w$, and $r_w$
remains unchanged.

Furthermore, because $a$ and $b$ both belong to stack $s$, exchanging them
does not change the stack from which the target at any position in the
retrieval order is retrieved. Consequently, for every target
$w\notin\{a,b\}$ and every stack $t$, the number $\kappa_t(w)$ of stack-$t$
targets retrieved before $w$ remains unchanged. It follows that every such
target remains feasible and that its contribution to the penalty is
unchanged.

\emph{Step 3 (Inversion counts of the exchanged targets).} Let $\rho$ denote
the number of targets $z$ in stack $s$ with $z>M$ that are retrieved before
$a$. We first compute the inversion counts of $a$ and $b$ under $\pi$. Since
$b\prec_\pi x$, every target $z>b$ in stack $s$ that is retrieved before $b$
also belongs to $B_x$, and therefore satisfies $z\le M$. Conversely, because
$b$ is the last target in $\{a+1,\ldots,M\}$ retrieved under $\pi$, every
target in $\{b+1,\ldots,M\}$ is retrieved before $b$. Hence, $r_b=M-b$.
Moreover, every target in $\{a+1,\ldots,M\}$ is retrieved before $x$, and
therefore before $a$. Together with the $\rho$ targets with index greater
than $M$, this gives $r_a=M-a+\rho$.

Let $r_a'$ and $r_b'$ denote the inversion counts after the exchange. At its
new position, $a$ is preceded by every target in
$\{a+1,\ldots,M\}\setminus\{b\}$. No target $z>M$ precedes $a$ at its
new position, since any such target would belong to $B_x$ under $\pi$,
contradicting $M=\max B_x$. Therefore, $r_a'=M-a-1$. At its new
position, $b$ is preceded by every target in $\{b+1,\ldots,M\}$ and by the
$\rho$ targets with index greater than $M$ that were retrieved before $a$.
Consequently, $r_b'=M-b+\rho$, and hence
$r_a'+r_b'=r_a+r_b-1$.

\emph{Step 4 (Feasibility after the exchange).} Fix a preceding stack $t<s$.
After the exchange, $a$ is retrieved at the former position of $b$, so the
number of stack-$t$ targets retrieved before $a$ is $\kappa_t(b)$. Similarly,
the number of stack-$t$ targets retrieved before $b$ at its new position is
$\kappa_t(a)$. Both counts are at most $m_t$. It therefore remains to verify
the corresponding cap bounds.

First, consider target $a$. Since $b\prec_\pi x$, we have
$\kappa_t(b)\le\kappa_t(x)$. The feasibility of $x$ gives
$\kappa_t(x)\le c_{x,t}+r_x$. Moreover,
$B_x\subseteq\{x+1,\ldots,M\}\setminus\{a\}$, and hence $r_x\le M-x-1$.
Applying inequality~\eqref{eq:gradient} to $x<a$ gives
$c_{x,t}\le c_{a,t}-(a-x)$. Consequently,
$\kappa_t(b)\le c_{a,t}+M-a-1=c_{a,t}+r_a'$, so $a$ is feasible at its new
position.

Now consider target $b$. The feasibility of $a$ under $\pi$ gives
$\kappa_t(a)\le c_{a,t}+r_a=c_{a,t}+M-a+\rho$. Applying
inequality~\eqref{eq:gradient} to $a<b$ gives
$c_{a,t}\le c_{b,t}-(b-a)$. Therefore,
$\kappa_t(a)\le c_{b,t}+M-b+\rho=c_{b,t}+r_b'$, so $b$ is also feasible at
its new position.

Since $t<s$ was arbitrary, both exchanged targets satisfy every feasibility
constraint under $\pi'$.

\emph{Step 5 (Change in energy and the inversion set of $x$).} By Step~2,
every target other than $a$ and $b$ retains the same contribution to the
penalty. After the exchange, $a$ has the values $\kappa_t(b)$, while $b$ has
the values $\kappa_t(a)$, for every preceding stack $t<s$. Because $a$ and
$b$ belong to the same stack, the same preceding stacks enter their penalty
terms. Hence, by~\eqref{eq:penalty-element}, their combined cross-stack
contribution is unchanged.

By Step~3, $r_a'+r_b'=r_a+r_b-1$. Each within-stack inversion in stack $s$
contributes $s$ to the penalty, so the combined within-stack contribution of
$a$ and $b$ decreases by exactly $s$. Therefore,
$P(\pi')=P(\pi)-s$. Because $E=K-\binom{n}{2}+P$
by~\eqref{eq:penalty-identity} and the remaining terms are
independent of the retrieval order, the same decrease holds for the energy.

Finally, the exchange changes the relative order of $x$ only with respect to
$a$ and $b$: target $b$ moves from before $x$ to after $x$, while $a$ moves
from after $x$ to before $x$. Therefore, $\pi'$ is feasible, satisfies
$E(\pi')=E(\pi)-s$, and the inversion set of $x$ under $\pi'$ is
$(B_x\setminus\{b\})\cup\{a\}$.
\end{proof}

\subsection{The Two-Stack Reduction and Decomposition}\label{app:m2}

\begin{lemma}[Monotonicity]\label{lem:mono}
For a fixed set of targets, increasing one or more caps cannot increase
$\OPT$. Likewise, replacing the window $[\lo,\hicmd]$ by
$[\lo',\hicmd']$, where $\lo'\le \lo\le \hicmd\le \hicmd'$,
cannot increase $\OPT$.
\end{lemma}
\begin{proof}
Each change only relaxes a feasibility constraint and leaves the objective
function unchanged. Hence, every solution that was feasible before the
change remains feasible afterward, so the optimum cannot increase.
\end{proof}

\begin{lemma}\label{lem:gradient}
Let $P_r\coloneqq[j{+}1..j{+}r]$. If the top target $j$ is feasible at
inversion count $r$, that is, $\lo\le\min(m_1,\hicmd,c_j+r)$,
then, for every $v\in[\lo,\hicmd]$, the targets in $P_r$ admit a feasible
solution with values in $[\lo,v]$.
\end{lemma}
\begin{proof}
When $r=0$, $P_r$ is empty and hence admits a feasible solution. Now suppose
that $r\ge 1$. Retrieve the targets in $P_r$ from bottom to top and assign
every target the value $\lo$. For $t\in\{1,\ldots,r\}$, target $j+t$ then
has inversion count $r-t$. Since the caps are strictly increasing integers,
$c_{j+t}\ge c_j+t$, and therefore
$c_{j+t}+(r-t)\ge c_j+r\ge\lo$. The latter inequality also gives
$\lo\le m_1$. Moreover, since
$v\ge\lo$, the assigned value $\lo$ lies in $[\lo,v]$. Thus, every target
satisfies its cap and window constraints, proving the claim.
\end{proof}

\begin{proof}[Proof of Theorem~{\normalfont\ref{thm:reduction-m2}}]
Consider any feasible retrieval order. Keep the positions occupied by stack
1 targets fixed and assign the targets to these positions from top to
bottom. The relative order of the stack 2 targets and the number of
preceding stack 1 retrievals for each such target remain the same. Thus,
their feasibility and all penalty terms involving them are unaffected.
Stack 1 targets have no preceding-stack feasibility constraints, while their
top-to-bottom order eliminates all within-stack penalties in stack 1. The
resulting order is therefore feasible and has no greater penalty. Hence, it
suffices to consider orders that retrieve the targets in stack 1 from top to
bottom.

Let $\pi$ be a feasible retrieval order with this property. Deleting the
stack 1 targets from $\pi$ gives a permutation $\sigma$ of $[q]$. For each
stack 2 target $j$, let $v_j\coloneqq\kappa_1(j)$ denote the number of stack
1 targets retrieved before $j$. Since this count cannot decrease during the
retrieval process, the values $v_j$ are nondecreasing along $\sigma$ and
satisfy $0\le v_j\le m_1$. Moreover, the targets below $j$ that are
retrieved before it are precisely the targets $j'>j$ that precede $j$ in
$\sigma$; their number is $r_j$. The feasibility
condition~\eqref{eq:retrieval-feasibility} therefore becomes
$v_j\le\min(m_1,c_j+r_j)$.
Thus, $(\sigma,v)$ satisfies the feasibility conditions of problem M2.

Since the stack 1 targets are retrieved from top to bottom, pairs of targets
in stack 1 contribute no penalty. Exactly $r_j$ lower targets from the same
stack are retrieved before each stack 2 target $j$, and each such pair
contributes $2$. Moreover, exactly $m_1-v_j$ stack 1 targets are retrieved
after $j$, and each forms a cross-stack pair that contributes $1$. Hence,
$P(\pi)=\sum_{j=1}^{q}\bigl(2r_j+m_1-v_j\bigr)=F(\sigma,v)$,
which equals the objective value of $(\sigma,v)$ in problem M2.

Conversely, consider a feasible solution $(\sigma,v)$ to problem M2. We
construct a corresponding retrieval order $\pi$. Now retrieve the stack 2
targets in the order $\sigma$. Before retrieving target $j$, retrieve
targets from stack 1, from top to bottom, until exactly $v_j$ stack 1 targets
have been retrieved. This is possible because the values $v_j$ are
nondecreasing along $\sigma$. After all stack 2 targets have been retrieved,
retrieve the remaining stack 1 targets from top to bottom.

Observe that every stack 1 target is accessible because there is no stack to
its left. Moreover, when stack 2 target $j$ is retrieved, we have
$\kappa_1(j)=v_j$, while its inversion count is $r_j$. Its feasibility
condition is therefore precisely the corresponding constraint of problem
M2. Consequently, $\pi$ is feasible, and Lemma~\ref{lem:penalty} yields
$P(\pi)=F(\sigma,v)$.

Thus, minimizing $P$ over feasible retrieval orders is equivalent to
minimizing $F$ over feasible solutions to problem M2. Finally,
by~\eqref{eq:instance-feasibility}, a two-stack instance is
feasible if and only if $n(1)\ge h_j-R(j)$ for every target $j$ in stack 2.
Since the targets are indexed from top to
bottom, $R(j)=q-j$. By $c_j=n(1)-h_j$, the feasibility condition is
equivalent to $c_j+(q-j)\ge 0$ for every $j\in[q]$.
\end{proof}

\stepcounter{lemma}
\begin{lemma}[Monotonicity of the upper bounds]\label{lem:deadlines}
For every $j\in[q]$, let
$D_j\coloneqq\min\bigl(m_1,c_j+(q-j)\bigr)$.
Then every feasible solution satisfies $v_j\le D_j$, and the sequence
$D_1,\ldots,D_q$ is nondecreasing.
\end{lemma}
\begin{proof}
Since $r_j\le q-j$, the feasibility condition of problem M2 gives
\[
v_j\le\min(m_1,c_j+r_j)
\le\min\bigl(m_1,c_j+(q-j)\bigr)=D_j.
\]
Moreover, since the caps are strictly increasing integers,
$c_{j+1}\ge c_j+1$ for every $j<q$. Consequently,
$c_{j+1}+\bigl(q-(j+1)\bigr)\ge c_j+(q-j)$.
Taking the minimum with $m_1$ preserves this inequality, and hence
$D_{j+1}\ge D_j$.
\end{proof}

\begin{proof}[Proof of Lemma~{\normalfont\ref{lem:grammar}}]
Suppose first that $\sigma$ is generated by the recursive rule in the
statement. We show by induction
on $q$ that every inversion set is a prefix. For $q\le 1$, every inversion
set is empty. An order with $q\ge 2$ has either the form
$\sigma=1\cdot\sigma_R$ or $\sigma=\sigma_L\cdot1\cdot\sigma_R$,
where, in the second case, $\sigma_L$ orders $[2..d]$ and $\sigma_R$ orders
$[d{+}1..q]$. Hence, $B_1=\emptyset$ or $B_1=[2..d]$. For every target in
$\sigma_L$ or $\sigma_R$, targets outside the respective order are retrieved
later or have smaller indices and therefore do not enter its inversion set.
The induction hypothesis shows that all remaining inversion sets are
prefixes.

Conversely, suppose that every inversion set is a prefix. We show by
induction on $q$ that $\sigma$ is generated by the recursive rule. For
$q\le 1$, the rule generates the only possible order. Since
$B_1$ is a prefix, either $B_1=\emptyset$ or $B_1=[2..d]$. In the first
case, $\sigma=1\cdot\sigma_R$. In the second,
$\sigma=\sigma_L\cdot1\cdot\sigma_R$,
where $\sigma_L$ orders $[2..d]$ and $\sigma_R$ orders $[d{+}1..q]$. Both
$\sigma_L$ and $\sigma_R$ retain the prefix property. By induction, they
are generated by the recursive rule. Thus, the same rule generates
$\sigma$.

Finally, let $j<i$. If $i>j+r_j$, the intervals are disjoint. Otherwise,
$i\in B_j$. If $i+r_i>j+r_j$, then $\ell\coloneqq j+r_j+1$ belongs to
$B_i$, so $\ell\prec_\sigma i\prec_\sigma j$,
which contradicts $\ell\notin B_j$. Hence, $i+r_i\le j+r_j$, and
$[i..i+r_i]\subseteq[j..j+r_j]$.
\end{proof}

\begin{proof}[Proof of Proposition~{\normalfont\ref{prop:decomp}}]
\emph{(a)} Retrieve the targets in $B$ according to the feasible solution
on $B$, followed by target $j$ with value $v$, and then the targets in
$L\setminus B$ according to the feasible solution on $L\setminus B$ with
increased caps. The targets in $B$ retain their inversion counts and values
and therefore remain feasible. Target $j$ has inversion count $r$ and is
feasible because $v\in[\lo,\min(m_1,\hicmd,c_j+r)]$.
For each target $s\in L\setminus B$, exactly the targets $b\in B$ with
$b>s$ are additional inverted successors. Its inversion count therefore
increases by $\lvert\{b\in B:b>s\}\rvert$,
which equals the increase in its cap. Hence, the feasibility condition for
$s$ under the increased cap is identical to its feasibility condition when
the targets in $B$ are retrieved before $s$. Moreover, all values assigned
to targets in $B$ are at most $v$, target $j$ has value $v$, and all values
assigned to targets in $L\setminus B$ are at least $v$. Thus, the values are
nondecreasing throughout the retrieval order, which is therefore feasible.

The solution on $B$ contributes $C_B$, and the solution on $L\setminus B$
contributes $C_R$. Every $b\in B$ forms a bottom-up pair with $j$,
contributing $2r$ in total, whereas pairs between $j$ and $L\setminus B$
contribute zero. Each pair $(b,s)\in B\times(L\setminus B)$ with $b>s$
contributes exactly $2$, giving $2\cross(B)$ in total. Target $j$
additionally contributes $m_1-v$. Together, these contributions give
$\Phi(B,v)$.

\emph{(b)} Conversely, restrict the given solution to $B$ and
$L\setminus B$. Every inverted successor of a target in $B$ that precedes
it also lies in $B$, so the restriction to $B$ is feasible. For each
$s\in L\setminus B$, the additional inverted successors from $B$ are
exactly the targets $b>s$, which are accounted for by its cap increase.
Since the values are nondecreasing and target $j$ has value $v$, the two
restrictions use values in $[\lo,v]$ and $[v,\hicmd]$, respectively. Their
costs, together with the contributions $m_1-v$, $2r$, and $2\cross(B)$,
give $\Phi(B,v)$.

If $B=P_r$, then $b<s$ for every $b\in B$ and $s\in L\setminus B$.
Hence, all cap increases and $\cross(B)$ are zero.
\end{proof}


\subsection{The Promotion Swap, Core Exchange, and the Two-Stack
Theorem}\label{app:core}

\begin{lemma}[Promotion swap]\label{lem:promswap}
Let $(\sigma,v)$ be any feasible solution whose inversion set $B_1$ of
target $1$ is not an index prefix. Put $M\coloneqq\max B_1$, let
\[
a\coloneqq\max\{x<M:\ x\notin B_1\}\qquad(\text{so } 2\le a\le M-1
\text{ and } (a..M]\coloneqq\{a{+}1,\dots,M\}\subseteq B_1),
\]
and let $b$ be the target in $(a..M]$ retrieved last under $\sigma$.
Exchange the positions of $a$ and $b$ in $\sigma$, leaving the value
assigned to each position unchanged. The result is a feasible solution
in which the inversion set of target $1$ is
$B_1\setminus\{b\}\cup\{a\}$ and the cost is exactly $2$ smaller.
\end{lemma}

\begin{proof}
Since $b\in B_1$ and $a\notin B_1$, we have $b\prec 1\prec a$. Moreover,
since $b$ is retrieved last under $\sigma$ among the targets in $(a..M]$,
every target in $(a..M]\setminus\{b\}$ precedes $b$. Let
$\rho\coloneqq\#\{z>M:\ z\prec a\}$.

\emph{Step 1 (indices of the intermediate targets).} Every target $w$
appearing strictly between $b$ and $a$ in $\sigma$ satisfies $w<a$ or
$w>M$. For this, first suppose that $w\prec 1$. Then $w\in B_1$. If
$w>a$, the definitions of $a$ and $M$ imply that $w\in(a..M]$. However,
$w$ is retrieved after $b$, whereas $b$ is the target in $(a..M]$
retrieved last under $\sigma$. Therefore, $w<a$. If $w=1$, then $w<a$.
Finally, if $1\prec w$, then $w\notin B_1\cup\{a\}$, and hence
$w\notin(a..M]\cup\{a\}$. It follows that $w<a$ or $w>M$.

\emph{Step 2 (all other targets are unaffected).} For any target
$w\notin\{a,b\}$, the exchange changes the set of targets preceding $w$
only if $w$ appears strictly between $b$ and $a$ in $\sigma$. In that
case, $b$ is replaced by $a$ among the targets preceding $w$. By Step~1,
either $w<a<b$ or $w>M\ge b$. Thus, $a>w$ if and only if $b>w$, and the
inversion count $r_w$, as well as the cap of $w$, is unchanged.
Every position retains its assigned value. Therefore,
the sequence of values along the resulting retrieval order remains
nondecreasing, and every target other than $a$ and $b$ retains its
assigned value. In particular, $v_1$ and $r_1$ are unchanged.

\emph{Step 3 (inversion counts of the moved targets).} Every target
retrieved before $b$ is also retrieved before target $1$ and hence belongs
to $B_1$. No target with an index greater than $M$ therefore precedes
$b$. Because $b$ is retrieved last under $\sigma$ among the
targets in $(a..M]$, every target in $(b..M]$ precedes $b$. Thus, the
targets with indices greater than $b$ that precede $b$ are exactly
$b+1,\ldots,M$, and $r_b=M-b$. Every target in $(a..M]$
also precedes $1$ and therefore precedes $a$, while $\rho$ counts the targets
with indices greater than $M$ that precede $a$. Hence,
$r_a=(M-a)+\rho$.

Let $r'_a$ and $r'_b$ denote the new inversion counts after the exchange.
Target $a$ is retrieved at the former position of $b$. Every target in
$(a..M]\setminus\{b\}$ precedes this position, whereas no target with an
index greater than $M$ does. Therefore, $r'_a=M-a-1$. Target $b$ is
retrieved at the former position of $a$. It is then preceded by the
targets $b+1,\ldots,M$ and by the $\rho$ targets with indices greater
than $M$ that preceded $a$. Therefore, $r'_b=(M-b)+\rho$.

\emph{Step 4 (feasibility at the earlier position).} At its new position,
$a$ is assigned $v_b\in[\lo,\hicmd]$. Since $b\prec 1$, monotonicity and
the feasibility of target $1$ give $v_b\le v_1\le c_1+r_1$. Moreover,
$B_1\subseteq\{2,\dots,M\}\setminus\{a\}$, so $r_1\le M-2$, while the
strictly increasing integer caps give $c_1\le c_a-(a-1)$. Hence,
\[
v_b\le c_a-(a-1)+(M-2)=c_a+(M-a-1)=c_a+r'_a,
\]
and $a$ is feasible.

\emph{Step 5 (feasibility at the later position).} At its new position,
$b$ is assigned $v_a\in[\lo,\hicmd]$. The feasibility of $a$ before the
exchange gives
\[
v_a\le c_a+r_a=c_a+(M-a)+\rho.
\]
Since the caps are strictly increasing integers, $c_a\le c_b-(b-a)$.
Hence,
\[
v_a\le c_b-(b-a)+(M-a)+\rho
=c_b+(M-b)+\rho
=c_b+r'_b,
\]
and $b$ is feasible.

\emph{Step 6 (cost and the new inversion set).} Since the value assigned
to each position is unchanged, the term $\sum_x(m_1-v_x)$ is unchanged.
By Step~2, only the inversion counts of $a$ and $b$ can change.
Step~3 gives $r_a+r_b=(M-a)+\rho+(M-b)$, whereas
$r'_a+r'_b=(M-a-1)+(M-b)+\rho$.
Thus, the total inversion count decreases by one and the cost decreases by
exactly $2$. Finally, $a$ now precedes target $1$, while $b$ follows it,
and no other target changes its position relative to target $1$.
Therefore, the new inversion set of target $1$ is
$B_1\setminus\{b\}\cup\{a\}$.
\end{proof}

Lemma~\ref{lem:promswap} immediately yields the following corollary.

\stepcounter{corollary}
\begin{corollary}[Inversion set of target $1$]\label{cor:topprefix}
In every optimal solution of a windowed instance, the inversion set $B_1$
of target $1$ is an index prefix.
\end{corollary}

\begin{corollary}[Recursive structure of optimal orders]\label{cor:grammar}
For every windowed instance, the order $\sigma$ of every optimal solution
$(\sigma,v)$ is a laminar-delay order.
\end{corollary}

\begin{proof}
Consider an optimal solution $(\sigma,v)$ on $q$ targets. If $q\le1$, its
order is trivially a laminar-delay order. Suppose $q\ge2$, and assume the
claim holds for all smaller windowed instances. By
Corollary~\ref{cor:topprefix}, the inversion set of target $1$ is
$B_1=P_r=[2..r{+}1]$. Hence,
$\sigma=\sigma_P\cdot 1\cdot\sigma_R$,
where $\sigma_P$ and $\sigma_R$ are the restrictions of $\sigma$ to $P_r$
and $[r{+}2..q]$, respectively. By Proposition~\ref{prop:decomp}(b), these
restrictions give feasible solutions to the subinstances
$(P_r,[\lo,v_1])$ and $([r{+}2..q],[v_1,\hicmd])$. Each solution is
optimal for its subinstance; otherwise, replacing it with a lower-cost
solution via Proposition~\ref{prop:decomp}(a) would contradict the
optimality of $(\sigma,v)$. By induction, $\sigma_P$ and $\sigma_R$ are
laminar-delay orders. If $r=0$, the production
\[
S([1..q])\longrightarrow 1\cdot S([2..q])
\]
generates $\sigma$. If $r>0$, the production
\[
S([1..q])\longrightarrow
S([2..r{+}1])\cdot 1\cdot S([r{+}2..q])
\]
generates $\sigma$. Thus, $\sigma$ is a laminar-delay order.
\end{proof}

\subsubsection*{The core exchange and the two-stack theorem}

\begin{proof}[Proof of Theorem~{\normalfont\ref{thm:core}}]
Observe first that if $\Phi(B,v)=+\infty$, the claim follows because the
left-hand side is finite for every admissible $v$, by Lemma~\ref{lem:gradient}
and the feasibility criterion for the remaining interval. Consider the
case $\Phi(B,v)<+\infty$. Choose optimal solutions to the two subproblems
defining $\Phi(B,v)$. By Proposition~\ref{prop:decomp}(a), combining these
solutions with target $j$, assigned value $v$, yields a feasible solution
$S_0$ for $(I,[\lo,\hicmd])$ in which the inversion set of $j$ is $B$ and
the cost is $\Phi(B,v)$.

After shifting the indices by $j-1$, Lemma~\ref{lem:promswap} applies to
target $j$. Whenever its inversion set differs from $P_r$, perform the
exchange from Lemma~\ref{lem:promswap}. Each exchange preserves
feasibility, the position and assigned value $v$ of $j$, and the
cardinality $r$ of its inversion set, while decreasing the cost by $2$.
Let $C$ denote the current inversion set of $j$. Each exchange replaces
some $b\in C$ by a target $a\notin C$ with $a<b$, thereby decreasing
$\cross(C)$ by $b-a\ge1$. Initially, $C=B$. Since $\cross(C)$ is a
nonnegative integer and equals zero precisely when $C=P_r$, after at most
$\cross(B)$ exchanges we obtain a feasible solution whose inversion set
for $j$ is $P_r$. Let $s$ denote the number of exchanges and let $S_s$
denote the resulting solution. Then $s\le\cross(B)$, the cost of $S_s$ is
$\Phi(B,v)-2s$, and $s\ge1$ if $B\ne P_r$.

Since the inversion set of $j$ in $S_s$ is $P_r$,
Proposition~\ref{prop:decomp}(b) shows that its restrictions to $P_r$ and
$[j{+}r{+}1..k]$ are feasible for the subinstances $(P_r,[\lo,v])$ and
$([j{+}r{+}1..k],[v,\hicmd])$, respectively. Their costs are at least the
corresponding optimal costs. Including the cost $(m_1-v)+2r$ of target
$j$, we obtain
\[
\OPT(P_r,[\lo,v])+(m_1-v)+2r+\OPT([j{+}r{+}1..k],[v,\hicmd])
\;\le\;\Phi(B,v)-2s.
\]
Since $s\ge0$, this establishes the first inequality. If $B\ne P_r$, then
$s\ge1$, and the right-hand side is at most $\Phi(B,v)-2$, establishing
the second inequality.
\end{proof}

\begin{proof}[Proof of Theorem~{\normalfont\ref{thm:main}}]
Consider an interval $I=[j..k]$ and a window $[\lo,\hicmd]$. If $I$ is
empty, both sides are zero. Suppose that $I$ is nonempty and that the
result holds for shorter intervals.

First, assume that $f(j,k,\lo,\hicmd)$ is finite and consider a minimizing
choice in the recurrence. By the induction hypothesis, its subproblem
values are attained by feasible solutions. Proposition~\ref{prop:decomp}(a)
combines these solutions with target $j$ into a feasible solution for
$(I,[\lo,\hicmd])$. Hence,
$\OPT(I,[\lo,\hicmd])\le f(j,k,\lo,\hicmd)$.

For the reverse inequality, assume that $\OPT(I,[\lo,\hicmd])$ is finite
and consider an optimal solution $(\sigma,v)$. Let $B$ be the inversion
set of $j$, let $r=|B|$, and let $v_j$ be the value assigned to $j$.
Proposition~\ref{prop:decomp}(b) decomposes this solution at $j$. Replacing
the two resulting solutions by optimal solutions to their respective
subproblems cannot increase the cost, and therefore
$\Phi(B,v_j)\le\OPT(I,[\lo,\hicmd])$.
By Theorem~\ref{thm:core}, replacing $B$ by $P_r$ does not increase the
cost. The resulting subproblems are shorter intervals, so their optimal
costs equal the corresponding dynamic-program values by the induction
hypothesis. This prefix decomposition is included in the recurrence, and
thus
\[
f(j,k,\lo,\hicmd)\le\Phi(B,v_j)\le\OPT(I,[\lo,\hicmd]).
\]
The two inequalities establish the result.
\end{proof}

\subsection{The Reduction to Problem M3}\label{app:m3reduction}


\begin{proof}[Proof of Theorem~{\normalfont\ref{thm:reduction}}]
\emph{Step 1 (Retrieving stack~1 from top to bottom).}
Consider a feasible retrieval order $\pi$. Keep the positions occupied by
targets from stack~1 fixed and assign these targets to the positions from
top to bottom. For every target $b$ in stack~2 or~3, the number $\kk_1(b)$
of targets retrieved earlier from stack~1 remains unchanged. Hence, its
feasibility and contribution to $P(\pi)$ are preserved. Targets in
stack~1 are always accessible because no stack lies to their left.
Moreover, retrieving them from top to bottom eliminates every within-stack
penalty in stack~1. The resulting order is therefore feasible and has no
greater penalty.

Consequently, it suffices to consider retrieval orders in which the targets
in stack~1 are retrieved from top to bottom. In such an order, stack~1
affects the retrievals from stacks~2 and~3 only through the number of its
targets that have already been retrieved. Observe that this number,
denoted by $v=\kk_1\in[0,m_1]$, is nondecreasing throughout the retrieval
process.

\emph{Step 2 (From a retrieval order to an M3 solution).}
Consider a feasible retrieval order $\pi$ in which the targets in stack~1
are retrieved from top to bottom. Deleting these targets yields an order
$\sigma$ of $S_2\cup S_3$. For each $b\in S_2\cup S_3$, set
$v_b\coloneqq\kk_1(b)$. These values belong to $[0,m_1]$ and are
nondecreasing along $\sigma$.

For $j\in S_2$, feasibility with respect to stack~1 gives
$v_j\leq\min\{m_1,\alpha_j+r_j\}$. For $k\in S_3$, feasibility with
respect to stacks~1 and~2 gives
$v_k\leq\min\{m_1,\beta_k+r_k\}$ and
$\kk_2(k)\leq\min\{q_2,\gamma_k+r_k\}$. Here $r_j$, $r_k$, and
$\kk_2(k)$ are determined by $\sigma$: they count, respectively, the lower
targets from $S_2$ retrieved before $j$, the lower targets from $S_3$
retrieved before $k$, and the targets from $S_2$ retrieved before $k$.
Thus, $(\sigma,v)$ is feasible for problem M3.

\emph{Step 3 (From an M3 solution to a retrieval order).}
Conversely, consider a feasible solution $(\sigma,v)$ to problem M3.
Process the targets in the order $\sigma$. Before retrieving a target $b$,
retrieve targets from stack~1, from top to bottom, until exactly $v_b$
targets from that stack have been retrieved. This is possible because the
values $v_b$ are nondecreasing along $\sigma$. After processing $\sigma$,
retrieve any remaining targets from stack~1.

The stack~1 retrievals are feasible because no stack lies to their left.
When $b$ is retrieved, $\kk_1(b)=v_b$, while $r_j$, $r_k$, and
$\kk_2(k)$ are precisely the values determined by $\sigma$. Therefore,
the constraints of problem M3 are exactly the feasibility conditions of
the constructed retrieval order. Hence, this order is feasible.

\emph{Step 4 (Preservation of the penalty).}
In an order that retrieves stack~1 from top to bottom, its targets
contribute no penalty. By Lemma~\ref{lem:penalty}, a target $j\in S_2$
contributes $2r_j+(m_1-v_j)$, while a target $k\in S_3$ contributes
$3r_k+(m_1-v_k)+(q_2-\kk_2(k))$. Hence,
\[
P(\pi)=\sum_{j\in S_2}\bigl[2r_j+(m_1-v_j)\bigr]
+\sum_{k\in S_3}\bigl[3r_k+(m_1-v_k)+(q_2-\kk_2(k))\bigr].
\]
This is exactly the objective function of problem M3. Steps~2 and~3 show
that a feasible solution in either formulation yields a feasible solution
in the other with the same penalty. Together with Step~1, this proves that
the minimum penalty over all feasible retrieval orders equals the minimum
cost of problem M3.
\end{proof}

\begin{proof}[Proof of Theorem~{\normalfont\ref{thm:real}}]
\emph{Necessity.} Consider a realizable cap vector
$(m_1,\alpha,\beta,\goff)$. Within each stack, target heights strictly
decrease with their top-to-bottom indices. Consequently, $\alpha$ and
$\beta$ are strictly increasing whenever they are nonempty, and the cap
vector is admissible. If $\alpha$ is nonempty, its first entry corresponds
to the highest target in stack~2. Since this target has height at most
$n(2)-1$, we have
$\alpha_1\geq n(1)-(n(2)-1)=1-\goff$, and hence
$\alpha_1+\goff\geq1$.

\emph{Sufficiency.} Consider an admissible cap vector
$(m_1,\alpha,\beta,\goff)$ satisfying $\alpha=\emptyset$ or
$\alpha_1+\goff\geq1$. Choose an integer $n(1)$ that is at least $1$,
$m_1$, $1-\goff$, and every entry of $\alpha$ and $\beta$, and set
$n(2)\coloneqq n(1)+\goff$. If $\beta$ is nonempty, set
$n(3)\coloneqq\max\{1,q_3,n(1)-\beta_1+1\}$; otherwise, set
$n(3)\coloneqq1$.

Place the targets in stack~1 at heights $0,\ldots,m_1-1$, each target
$j\in S_2$ at height $n(1)-\alpha_j$, and each target $k\in S_3$ at
height $n(1)-\beta_k$. All these heights are nonnegative by the choice of
$n(1)$. The stack~1 targets fit because $m_1\leq n(1)$. If $\alpha$ is
nonempty, the stack~2 heights are at most
$n(1)-\alpha_1\leq n(2)-1$. Similarly, if $\beta$ is nonempty, the
stack~3 heights are at most $n(1)-\beta_1\leq n(3)-1$. Strict
monotonicity of $\alpha$ and $\beta$ ensures that the target heights
within each stack are distinct. These placements therefore define a
three-stack instance with cap vector $(m_1,\alpha,\beta,\goff)$, proving
realizability.
\end{proof}

\subsection{Dynamic-Program Soundness and the Recursive-Decomposition
Minimum}\label{app:dp}

\begin{definition}[recursive decomposition]\label{def:productions}
Let $\sigma$ be an order on targets from $S_2\cup S_3$. The order is
recursively decomposable if it contains at most one target or admits
one of the following decompositions.

Choose a stack $s\in\{2,3\}$ represented in $\sigma$, and let $(s,j_s)$ be its
uppermost remaining target. This target may appear first in $\sigma$, in which
case the suffix following it must be recursively decomposable. Alternatively,
it may be delayed behind a block containing exactly the next $b\geq 1$ targets
of stack $s$, namely $(s,j_s{+}1),\ldots,(s,j_s{+}b)$, and the first $r\geq 0$
remaining targets of the other stack $t$, namely
$(t,j_t),\ldots,(t,j_t{+}r{-}1)$, where $(t,j_t)$ is that stack's uppermost
remaining target. The latter sequence is empty when $r=0$ or stack $t$ is
absent. The two sequences may be interleaved within the block.
Both the block and the suffix following $(s,j_s)$ must themselves be
recursively decomposable.
\end{definition}

\begin{lemma}\label{lem:book}
Consider a recursive decomposition at the uppermost remaining target
$(s,j_s)$, where $s\in\{2,3\}$, and let $t$ denote the other stack. In
stack $s$, the block together with the selected target contains a consecutive
sequence of targets beginning with index $j_s$. In stack $t$, the block is
empty or contains a consecutive sequence beginning with its uppermost
remaining target. If the selected target is retrieved first, the block is
empty. The suffix contains exactly the remaining targets of the current
subproblem. Consequently, starting from $S_2$ and $S_3$, every subproblem in
the recursive decomposition contains one index interval from each stack,
possibly empty.
\end{lemma}

\begin{proof}[Proof of Lemma~{\normalfont\ref{lem:book}}]
In the delayed case, Definition~\ref{def:productions} requires the block to
contain consecutive targets immediately below $(s,j_s)$ in stack $s$ and,
if any, consecutive targets beginning with the uppermost remaining target of
stack $t$. Together with the selected target, these form an initial index
interval in each stack. If the selected target is retrieved first, only that
target is removed. In both cases, the suffix contains exactly the remaining
targets. Thus, whenever the current subproblem contains one index interval
from each stack, so do the block and suffix. Starting from $[1..q_2]$ and
$[1..q_3]$, the claim follows by induction.
\end{proof}

\begin{proof}[Proof of Proposition~{\normalfont\ref{prop:assembly}}]
Consider a state $(j_2,k_2,j_3,k_3,\lo,\hicmd,B)$ in the recursive
construction of a complete retrieval order. Before the targets in this
subproblem are retrieved, exactly $B$ stack-2 targets have been retrieved,
all with indices smaller than $j_2$. Similarly, every previously retrieved
stack-3 target has index smaller than $j_3$. Any target with a smaller index
in its stack that remains unretrieved has been delayed until after the
current subproblem. All previously retrieved targets have assigned values
at most $\lo$.

Suppose that $f(j_2,k_2,j_3,k_3,\lo,\hicmd,B)$ is finite. We show by induction
on the number of targets in the two intervals that the recurrence constructs
a retrieval order $\sigma'$ containing exactly these targets, with
nondecreasing values in $[\lo,\hicmd]$. In the complete retrieval order:
\begin{enumerate}
\item Each target in $\sigma'$ has the same within-stack inversion count as
in $\sigma'$.
\item For each stack-3 target $k$ in $\sigma'$, $\kk_2(k)$ equals $B$ plus
the number of stack-2 targets preceding $k$ in $\sigma'$.
\item Every target in $\sigma'$ satisfies the feasibility conditions of
problem M3.
\item The targets in $\sigma'$ contribute exactly
$f(j_2,k_2,j_3,k_3,\lo,\hicmd,B)$ to the M3 objective.
\end{enumerate}

The claim is immediate when both intervals are empty. Otherwise, the
recurrence selects one of cases \textnormal{(I2)}, \textnormal{(I3)},
\textnormal{(D2)}, or \textnormal{(D3)}. We verify these cases below.

\emph{Cases (I2) and (I3).} Both cases retrieve the uppermost remaining
target of one stack first and assign it value $v$. Every previously retrieved
target in the same stack has a smaller index, while all lower targets are
retrieved later. Thus, the selected target has within-stack inversion count
zero.

In case \textnormal{(I2)}, the condition $v\leq\alpha_{j_2}$ ensures
feasibility of target $(2,j_2)$. Its contribution to the M3 objective is
$m_1-v$. After its retrieval, the number of previously retrieved stack-2
targets increases to $B+1$, and the recurrence continues with
$f(j_2{+}1,k_2,j_3,k_3,v,\hicmd,B{+}1)$.

In case \textnormal{(I3)}, target $(3,j_3)$ is preceded by exactly $B$
stack-2 targets, so $\kk_2(j_3)=B$. The conditions $v\leq\beta_{j_3}$ and
$B\leq\min(q_2,\gamma_{j_3})$ ensure its feasibility. Its contribution to
the M3 objective is $(m_1-v)+(q_2-B)$, and the recurrence continues with
$f(j_2,k_2,j_3{+}1,k_3,v,\hicmd,B)$.

In both cases, the remaining targets form a smaller subproblem with value
window $[v,\hicmd]$. By the induction hypothesis, the corresponding state
supplies a retrieval order satisfying the four properties. Retrieving the
selected target first and then following this order establishes the
properties for the current state.

\emph{Case (D2).} Target $(2,j_2)$ is delayed behind a block
containing the stack-2 targets indexed by $[j_2{+}1..d_2]$ and the stack-3
targets indexed by $[j_3..d_3]$, and is assigned value $v$. By the induction
hypothesis, $f(j_2{+}1,d_2,j_3,d_3,\lo,v,B)$ supplies a retrieval order for
this block with nondecreasing values in $[\lo,v]$. Target $(2,j_2)$ remains
unretrieved while the block is processed.

The selected target is then retrieved at value $v$. Its inverted successors are
exactly the $d_2-j_2$ stack-2 targets in the block, so
$r_{j_2}=d_2-j_2$. The condition
$v\leq\alpha_{j_2}+(d_2-j_2)$ ensures its feasibility. Its contribution to
the M3 objective is
$2(d_2-j_2)+(m_1-v)$.

The block and the selected target contain $d_2-j_2+1$ stack-2 targets. Thus,
before the suffix begins, $B+(d_2-j_2+1)$ stack-2 targets have been retrieved.
By Lemma~\ref{lem:book}, the suffix contains the intervals
$[d_2{+}1..k_2]$ and $[d_3{+}1..k_3]$. The induction hypothesis supplies a
retrieval order for this subproblem with nondecreasing values in
$[v,\hicmd]$. Retrieving the block, followed by $(2,j_2)$ and then the
suffix, establishes the four properties. Since case \textnormal{(D2)}
attains the minimum, the total contribution equals
$f(j_2,k_2,j_3,k_3,\lo,\hicmd,B)$.

\emph{Case (D3).} Target $(3,j_3)$ is delayed behind a block containing the
stack-2 targets indexed by $[j_2..d_2]$ and the stack-3 targets indexed by
$[j_3{+}1..d_3]$, and is assigned value $v$. By the induction hypothesis,
$f(j_2,d_2,j_3{+}1,d_3,\lo,v,B)$ supplies a retrieval order for this block
with nondecreasing values in $[\lo,v]$. Target $(3,j_3)$ remains unretrieved
while the block is processed.

The selected target is then retrieved at value $v$. Its inversion count is
$r_{j_3}=d_3-j_3$, and it is preceded by $K_2=B+(d_2-j_2+1)$ stack-2
targets. The conditions $v\leq\beta_{j_3}+(d_3-j_3)$ and
$K_2\leq\min(q_2,\gamma_{j_3}+(d_3-j_3))$ ensure its feasibility. Its
contribution to the M3 objective is $3(d_3-j_3)+(m_1-v)+(q_2-K_2)$.

Retrieving $(3,j_3)$ leaves the number of retrieved stack-2 targets at
$K_2$. By Lemma~\ref{lem:book}, the suffix contains the intervals
$[d_2{+}1..k_2]$ and $[d_3{+}1..k_3]$. The induction hypothesis supplies a
retrieval order for this subproblem with nondecreasing values in
$[v,\hicmd]$. Retrieving the block, followed by $(3,j_3)$ and then the
suffix, establishes the four properties. Since case \textnormal{(D3)}
attains the minimum, the total contribution equals
$f(j_2,k_2,j_3,k_3,\lo,\hicmd,B)$.

At the initial state, no target has been retrieved, $B=0$, and the value
window is $[0,m_1]$. The construction therefore yields a feasible retrieval
order of cost $f(1,q_2,1,q_3,0,m_1,0)$, proving
$\OPT\leq f(1,q_2,1,q_3,0,m_1,0)$.
\end{proof}

\begin{proof}[Proof of Proposition~{\normalfont\ref{prop:parsing-min}}]
By Proposition~\ref{prop:assembly}, every finite initial-state value is
attained by a feasible retrieval order constructed through cases
\textnormal{(I2)}, \textnormal{(I3)}, \textnormal{(D2)}, and
\textnormal{(D3)}. These cases follow the recursive decomposition in
Definition~\ref{def:productions}, so the constructed order is recursively
decomposable. The dynamic-program value is therefore no smaller than the
minimum cost among feasible solutions with recursively decomposable orders.
If no such solution exists, the initial-state value must be $+\infty$.

For the reverse inequality, consider a feasible solution $(\sigma,v)$ to
problem M3 with a recursively decomposable order, and fix a decomposition
under Definition~\ref{def:productions}. Let $\sigma'$ be a consecutive
subsequence arising in this decomposition. By Lemma~\ref{lem:book}, its
targets form one index interval from each stack. These intervals define a
dynamic-program state, with $B$ equal to the number of stack-2 targets
retrieved before $\sigma'$. The initial state has value window $[0,m_1]$.
If the selected target receives value $v^*$, the suffix receives window
$[v^*,\hicmd]$ and, in the delayed case, the block receives window
$[\lo,v^*]$. Since the values are nondecreasing along $\sigma$, every
target lies within its assigned value window.

By Lemma~\ref{lem:book}, previously retrieved targets have smaller indices
in their respective stacks, while lower targets outside $\sigma'$ are
retrieved afterward. Thus, each target's within-stack inversion count in
the complete order equals its inversion count in $\sigma'$. For each
stack-3 target $k$ in $\sigma'$, $\kappa_2(k)$ equals $B$ plus the number
of stack-2 targets preceding $k$ in $\sigma'$.

We show by induction on the number of targets in $\sigma'$ that the
corresponding state value does not exceed their total contribution to the
M3 objective. The claim is immediate when $\sigma'$ is empty.

For a nonempty $\sigma'$, choose the recurrence case corresponding to its
fixed decomposition and assign the selected target its value $v^*$. The
block, when present, and the suffix have exactly the index intervals, value
windows, and numbers of previously retrieved stack-2 targets specified by
this case. The counts established above and feasibility of $(\sigma,v)$
ensure that all conditions of the case hold. The recurrence adds exactly
the selected target's contribution to the M3 objective. By the induction
hypothesis, the block and suffix state values do not exceed their respective
contributions. Their sum, including the selected target's contribution,
therefore bounds the current state value from above.

At the initial state,
$f(1,q_2,1,q_3,0,m_1,0)$ is therefore at most the cost of $(\sigma,v)$.
Minimizing over all feasible solutions with recursively decomposable
orders proves the reverse inequality.
\end{proof}

\end{APPENDICES}

\end{document}